\documentclass[11pt, a4paper]{amsart}
\usepackage{a4wide}
\usepackage[utf8]{inputenc}
\usepackage[english]{babel}
\usepackage{amsthm}

\newtheorem{theorem}{Theorem}[section]
\newtheorem{proposition}{Proposition}[section]
\newtheorem{corollary}{Corollary}[section]

\newtheorem{assumption}{Assumption}
\usepackage{amsfonts}
\usepackage{amssymb}
\usepackage{mathtools}
\usepackage{mathrsfs}
\usepackage[eulergreek]{sansmath}
\usepackage{amsmath}
\usepackage{physics}
\usepackage{subcaption}
\usepackage{chemformula}
\usepackage{chemfig}
\usepackage{graphicx}
\usepackage{acro}
\DeclareAcronym{VAE}{
  short = VAE,
  long = variational auto-encoder}
\DeclareAcronym{TEM}{
  short = TEM,
  long = transmission electron microscope}
\DeclareAcronym{Cryo-SPA}{
  short = Cryo-SPA,
  long = single-particle cryo–electron microscopy}
\DeclareAcronym{Cryo-EM}{
  short = Cryo-EM,
  long = cryogenic electron microscopy}
\DeclareAcronym{LDDMM}{
  short = LDDMM,
  long = large deformation diffeomorphic metric mapping}
\DeclareAcronym{BFGS}{
  short = BFGS,
  long = Broyden–Fletcher–Goldfarb–Shanno}
\DeclareAcronym{ML-EM}{
  short = ML-EM,
  long = expectation-maximization maximum-likelihood}
\DeclareAcronym{MAP}{
  short = MAP,
  long = maximum a posteriori}
\DeclareAcronym{CTF}{
  short = CTF,
  long = contrast transfer function}  
\usepackage{hyperref}
\hypersetup{%
    unicode=true,     
    colorlinks=true, 
    linkcolor=blue,   
    citecolor=red, 
    filecolor=blue,   
    urlcolor=blue
}
\usepackage[nameinlink]{cleveref}
\crefname{assumption}{assumption}{assumptions}
\Crefname{assumption}{Assumption}{Assumptions}
\crefalias{AlgoLine}{line}%
\makeatletter
\let\cref@old@stepcounter\stepcounter
\def\stepcounter#1{%
  \cref@old@stepcounter{#1}%
  \cref@constructprefix{#1}{\cref@result}%
  \@ifundefined{cref@#1@alias}%
    {\def\@tempa{#1}}%
    {\def\@tempa{\csname cref@#1@alias\endcsname}}%
  \protected@edef\cref@currentlabel{%
    [\@tempa][\arabic{#1}][\cref@result]%
    \csname p@#1\endcsname\csname the#1\endcsname}}
\makeatother
\usepackage{etoolbox}

\newcommand{\Cdot}{\,\cdot\,}
\newcommand{\dint}{\,\mathrm{d}}
\makeatletter
\newcommand{\subalign}[1]{%
  \vcenter{%
    \Let@ \restore@math@cr \default@tag
    \baselineskip\fontdimen10 \scriptfont\tw@
    \advance\baselineskip\fontdimen12 \scriptfont\tw@
    \lineskip\thr@@\fontdimen8 \scriptfont\thr@@
    \lineskiplimit\lineskip
    \ialign{\hfil$\m@th\scriptstyle##$&$\m@th\scriptstyle{}##$\hfil\crcr
      #1\crcr
    }%
  }%
}
\makeatother

\newcommand{\Real}{\mathbb{R}}

\newcommand{\opStyle}[1]{\operatorname{\mathcal{#1}}}
\newcommand{\argmin}{\operatorname*{arg\,min}}
\newcommand{\inpro}[3][{}]{ \left( #2 , #3 \right)_{#1}}
\newcommand{\pair}[3][{}]{\left\langle #2, #3 \right\rangle_{#1}}
\newcommand{\traceop}{\operatorname{tr}}
\newcommand{\leftact}[1][{}]{L_{#1}}
\newcommand{\rightact}[1][{}]{R_{#1}}
\newcommand{\orb}[1][{}]{\mathcal{O}_{#1}}
\newcommand{\LpSpace}{\mathscr{L}}

\newcommand{\LieGroup}{G}
\newcommand{\LieGroupDistance}{\operatorname{d}_{\LieGroup}}

\newcommand{\gelem}{g}
\newcommand{\gelemother}{h}
\newcommand{\gcurve}{\gamma}
\newcommand{\LieAlgebra}{\mathfrak{g}}

\newcommand{\aelem}{u}
\newcommand{\aelemother}{\tilde{u}}

\newcommand{\ShapeSpace}{V}
\newcommand{\template}{w}

\newcommand{\target}{y}
\newcommand{\GroupAction}{\Phi}
\newcommand{\DataSpace}{Y}
\newcommand{\data}{y}
\newcommand{\dataother}{\tilde{y}}

\newcommand{\ForwardOp}{\opStyle{F}}

\newcommand{\RegFunc}{\opStyle{S}}
\newcommand{\LossFunc}{\opStyle{L}}
\newcommand{\bbpos}{a}

\newcommand{\relbbpos}{\amod}
\newcommand{\relbbposother}{\amodother}
\newcommand{\numPep}{N}
\newcommand{\numImgs}{m}
\newcommand{\MapSpace}{X}

\newcommand{\AtmModSpace}{A}
\newcommand{\amod}{v}
\newcommand{\amodother}{\amod'}

\newcommand{\AtmToMapOp}{\opStyle{B}}
\newcommand{\AtmToMapOpPlane}{\opStyle{B}_{2D}}
\newcommand{\amodplane}{p}

\newcommand{\SO}{\operatorname{SO}}
\newcommand{\SOLieAlgebra}{\mathfrak{so}}

\newcommand{\rotmat}{\rho}

\newcommand{\gauss}{\varphi}

\newcommand{\inertia}{\mathbb{A}}

\usepackage{tabularray}
\usepackage{siunitx}
\UseTblrLibrary{booktabs, siunitx}

\newtoggle{showrevision}
\togglefalse{showrevision} 
\iftoggle{showrevision}{%
  
}{%
 
}

\title[Indirect shape gradient flows for cryo-EM]{Recovering protein conformations from single-particle cryo-EM data via indirect shape matching gradient flows}
\author{Erik Jansson$^{*}$} \address{Department of Applied Mathematics and Theoretical Physics, University of Cambridge, United Kingdom}
\email{eoj23@cam.ac.uk}
\thanks{$^*$ Corresponding author. }
\author{ Jonathan Krook}\address{Department of  Mathematics, Royal Institute of Technology, Sweden}
\author{Ozan Öktem} \address{Department of Mathematics, Royal Institute of Technology, Sweden}
\author{Carola-Bibiane Schönlieb} \address{Department of Applied Mathematics and Theoretical Physics, University of Cambridge, United Kingdom}
\begin{document}

\begin{abstract}
\Ac{Cryo-SPA} images a macromolecule as many noisy tomographic projections of its electrostatic potential. We reconstruct the protein backbone directly from such projections, as an atomic point cloud, without the intermediate step of reconstructing the 3D electrostatic potential map. We formulate this
as an indirect shape matching problem: a point-cloud template of the backbone is deformed until its simulated projections agree with the data, with the structure observed only through the imaging operator. The deformation is computed via a gradient flow on a Lie group, and we derive the resulting framework in a general geometric setting before adapting it for \ac{Cryo-SPA}. On synthetic data we recover
single- and multi-chain proteins and capture conformational transitions.
\end{abstract}



\maketitle
\acresetall

\section{Introduction}
\label{sec:intro}

An essential part of biomedical research is to study the functionality of biomolecules, such as proteins. 
The dynamics of the 3D shape, or structure, of a protein have a large influence on its functionality. 
Naturally, a goal of structural biology is to develop methods, both experimental and computational, for determining and studying the structure  of proteins. 
Notable among these methods are those based on \ac{TEM} imaging, an idea that goes back to Rosier and Klug \cite{Rosier:1968aa}, who estimated the 3D structure of the bacteriophage T4 tail from electron micrographs. 

A key difficulty is to obtain the necessary amount of 2D \ac{TEM} images of the biomolecule from sufficiently many different views.
This is achievable when using vitrified specimens that consist of several isolated and structurally identical copies of the biomolecule in varying orientations. \Ac{Cryo-SPA} is a 3D electron microscopy technique designed for such data.
It has since its introduction in \cite{Frank:1975aa,Heel:1981aa,Vainstein:1986aa,Heel:1987aa} undergone a rapid development especially regrading computational methods,
\cite{Bai:2015aa,Nogales:2015aa,Cheng:2015aa,Carazo:2015aa,Sigworth:2016aa,Cheng:2018aa,Singer:2018aa,Singer:2020aa,Bendory:2020aa,Benji:2020aa}, sample preparation \cite{Liu:2023aa,Venien-Bryan:2023aa,Cheng:2024aa}, and instrumentation \cite{Faruqi:2015aa,Williams:2019aa}.
Through this development, which was also recognized by the 2017 Nobel Prize in Chemistry, \ac{Cryo-SPA} has become a major tool in structural biology for studying large rigid biomolecules (homogeneous particles) at (near) atomic resolution \cite{Cheng:2015ab,Nakane2020,Herzik:2020aa} that are difficult or impossible to crystallize and has proven essential in applications such as life science \cite{Nogales:2016aa} and drug discovery \cite{Renaud:2018aa,Robertson:2022aa}.

\subsection{Basics of \ac{Cryo-SPA}}
The starting point in \ac{Cryo-SPA} is a cryofixated specimen consisting of a thin slab of amorphous ice, containing several \emph{particles}, that is, instances of the same biomolecule.
This technique was introduced in \cite{Dubochet:1982aa,Adrian:1984aa} to solidify a specimen, which is needed as \ac{TEM} imaging takes place in vacuum, while preserving its structural integrity.

The vitreous slab representing the cryofixated specimen is then imaged in a \ac{TEM} where high-energy electrons scatter against the specimen, resulting in a few large 2D phase contrast \ac{TEM} images (micrographs). 
Particle picking is then used to computationally extract 2D subregions (particle images) from the micrographs, each representing a noisy 2D ``projection image'' of a \emph{single} particle.
All particles in the specimen represent the \emph{same} biomolecule, but their 3D shapes can differ depending on the specific conformation of the biomolecule that the particle represents.
Mathematically, we view particle images as elements $\data_1, \ldots, \data_\numImgs \in \DataSpace$ where the data space $\DataSpace$  is the vector space of $\Real$-valued $\LpSpace^2$-functions on $\Real^2$.

\subsection{3D reconstruction and model building}
A key part of \ac{Cryo-SPA} is 3D reconstruction, which refers to the task of recovering the volumetric electrostatic potential, henceforth called the 3D map, that is generated by each particle from the corresponding 2D particle image (one such image per particle). 
A challenge is that there is only a single 2D particle image associated with each particle, so 3D reconstruction amounts to recovering a 3D function from a single 2D image. This is only possible with clever use of the fact that all particles represent the same biomolecule, this even if they do not have identical 3D maps.  
Another challenge is that the particle has an unknown 3D orientation (pose) when generating the 2D particle image. 
Next, 2D particle images have very low signal-to-noise ratio. One reason for this is that the vitrified biological specimens result in images with very low contrast, as they act as weak phase contrast objects when imaged in a \ac{TEM}. 
Furthermore, biological specimens are sensitive to damage from radiation, so \ac{TEM} images need to be obtained using low doses (200--400~$\text{electrons}/\text{nm}^2$).

Another part of \ac{Cryo-SPA} is model building, which is the task of building (pseudo) atomic models for the particles.
Since interpretation and identification of biological functionality often rely on having access to (pseudo) atomic models, model building is essential.
In the setting we consider here, of single-chain proteins, model building is the task of  ``folding'' the known primary structure of the protein into the particle specific 3D map.

\subsection{Mathematical shape analysis}
In Jansson et al. \cite{Jansson2025}, we showed how to use shape analysis for jointly performing 3D reconstruction and model building in \ac{Cryo-SPA}.

Shape analysis is the task of finding an energy-minimizing way to \emph{deform} an initial shape (template) into a target shape. 
The modern mathematical foundations of such a framework were laid by by Grenander \cite{Gr1993,grenander2007} and it has since then found applications in medical image analysis (computational anatomy), see for instance \cite{Ceritoglu2013,Risser2013}.
Within this framework, new shapes are mathematically obtained by deforming a template through the action of a Lie group. 
The shapes, including the template, are elements in a ``shape space'', which can be, e.g., spaces of points, curves, surfaces or functions. 
\emph{Shape matching} is then reduced to the task of finding a Lie group element that through its group action maps the template shape as close as possible to a target shape.
Closeness is here defined by having a small \emph{matching energy}, so shape matching translates to solving an optimization problem on Lie groups. 
The reader may consult the book by Younes~\cite{Younes2010} and references therein for an overview of mathematical shape matching. 

Computational shape analysis is typically approached by taking the Lie group element deforming the template as the endpoint of a curve in the Lie group that in turn generates a curve in the Lie algebra.
This latter curve is governed by a set of equations known as the \emph{Euler--Poincaré} equations. 
Optimization therefore proceeds as follows: Starting from an initial value for the Lie algebra curve, one solves the Euler–Poincaré and flow equations to generate a Lie group element, which is then used to evaluate the matching energy.
The initial guess is then updated with the gradient of the energy with respect to this initial value of the Lie algebra curve.
This approach is, however, prohibitively slow, which in turn limits its usefulness for joint 3D reconstruction and model building in \ac{Cryo-SPA}.

\subsection{Specific contributions}
In this paper, we consider an alternative approach to shape analysis for joint 3D reconstruction and model building in \ac{Cryo-SPA}. It originates in work by Balehowsky et al. \cite{BaKaMo2022} that proposed a different approach, namely to compute a minimizer by following a gradient flow formulated directly on the Lie group. 

By exploiting the geometric structure of the group, the gradient flow is straightforward to express without any algebra-to-group maps. 
Empirical observation suggests that the resulting flows offer significant speed up compared to the approach used in Jansson et al. \cite{Jansson2025}. 
This opens up the possibility of using shape-based reconstruction in more complicated settings, for instance, for multi-chain proteins and for \ac{Cryo-SPA} where particles display continuous heterogeneity. 
Our main contribution is twofold:
\begin{enumerate}
    \item Extension of the geometric framework for gradient flow-based shape matching to the indirect shape matching problem, i.e., for reconstruction, a general framework that can also be applied to other problems, for instance, tomographic imaging.
    \item The application of this framework to the problem of estimating a (pseudo) atomic model of a biomolecule from \ac{Cryo-SPA} particle images, i.e., to perform joint 3D map estimation and model building 
\end{enumerate}

We next place our approach into the larger landscape of computational methods for \ac{Cryo-SPA}. 
There are many approaches to solving the reconstruction problem in \ac{Cryo-SPA}.
We organize the discussion around the distinction our method turns on: whether the atomic model is obtained sequentially, by first reconstructing a 3D map and then fitting an atomic model into it, or jointly. 
General background is given in the survey of Bendory et al.~\cite{Bendory:2020aa}, and reconstruction methods when particles display continuous conformational heterogeneity are surveyed in Sorzano et al.~\cite{Sorzano:2019aa}.

Most methods only target reconstructing a 3D map. 
For a biomolecule with a single conformation, the particles have identical 3D  map (homogeneous particles).
Such a 3D map can be computed with  RELION~\cite{Scheres:2012ab}, that uses expectation maximization to compute a maximum a posteriori estimator that is marginalized over poses, and cryoSPARC~\cite{Punjani:2017aa}, which uses branch and bound for pose estimation instead of marginalization.
When several conformations are present, the particles do not have identical 3D maps (heterogeneous particles).
RECOVAR~\cite{gilles2025} treats continuous heterogeneity through a principal component analysis, and CryoDRGN~\cite{Zhong:2021aa} through a neural network that maps a conformational latent variable to a 3D map.
In all of these, recovering an atomic model requires the further, non-trivial step of fitting the model into the reconstructed 3D map.

We instead reconstruct the atomic model of the protein backbone directly.
The backbone is a set of \ch{C_{$\alpha$}} positions, obtained by deforming a template backbone through the action of a deformation group. To compare a candidate backbone with the data, each \ch{C_{$\alpha$}} is represented by a Gaussian, and the resulting density is projected through the forward model; the deformation is driven so that these projections match the observed images. The reconstruction is therefore a set of atomic coordinates, produced in a single stage, without the intermediate step of reconstructing a 3D map and fitting a model into it. Our examples are primarily homogeneous, but we also treat a simplified heterogeneous setting, in which the images are grouped by a conformational coordinate estimated from the data, and each group is reconstructed separately.
Further, we do not consider the orientation estimation problem, that is, to determine global three-dimensional orientation of the protein given only the noisy images, and note that approaches to solve it exists, see e.g., \cite{Diepeveen:2023aa} and the references therein. 

Reconstructing an atomic model directly by projecting Gaussians has been considered before.
E2GMM~\cite{Chen2021} learns a Gaussian mixture representing the density with a neural network, without imposing backbone connectivity.
Esteve-Yag\"ue et al.~\cite{EsteveYage2023} also project Gaussians and then recover the backbone geometry in a second stage by estimating bond and torsion angles with manifold learning.
Our reconstruction differs in that it is a single variational procedure that deforms a template backbone by a gradient flow on a group of deformations. 
The output is a backbone model by construction and no separate geometry-estimation stage is required.

The deformation is driven by a gradient flow on a group of deformations, which performs shape registration. Since the target is available only through the 2D particle image, and not as a given shape, the registration is indirect.
This extends the direct registration of Balehowsky et al.~\cite{BaKaMo2022}, who study gradient flows on deformation groups for registration between given shapes and argue for their lower computational cost relative to \ac{LDDMM}. We previously used an \ac{LDDMM} approach in the \ac{Cryo-EM} setting~\cite{Jansson2025}, but  here we use the gradient-flow formulation instead, for the same reason of reduced computational complexity.

\section{Indirect shape matching problems}
Formally, shape matching refers to the task of aligning selected features of a template object with corresponding features of an indirectly observed target.
The starting point is to consider the set $\ShapeSpace$ of deformable shapes (shape space). $\ShapeSpace$ is often a vector space, but can be more general, e.g., a manifold. 

Elements of $\ShapeSpace$  are deformed by the  action of a group $\LieGroup$ on $\ShapeSpace$.
Mathematically, we denote the  group action by
$\GroupAction \colon \LieGroup \times \ShapeSpace \to \ShapeSpace$, and it represents deformations of elements in $\ShapeSpace$ that are parametrized by elements in $\LieGroup$.
Additionally, we assume that $\LieGroup$ has a manifold structure where the group operations are smooth, i.e., $\LieGroup$ is a Lie group. 
We denote its corresponding Lie algebra  by $\LieAlgebra$.
In the following, we denote  the group identity element by $e$.

In the typical treatment of shape matching (see e.g., \cite{Bruveris2013, Younes2010}), the space of deformable objects can be, for instance, spaces of point clouds (also known as landmarks), curves, surfaces, 3D densities, smooth functions, or currents. 
The Lie group is then usually the infinite-dimensional group of diffeomorphisms. 
In this work, we focus on a setting where the Lie group is a finite-dimensional manifold, but remark that the geometric framework is also applicable in the infinite-dimensional setting. 

Here, we adopt the approach of indirect matching. Concretely, this is the problem of matching a template $\template \in \ShapeSpace$ against a target that is indirectly observed through (noisy) data $\data \in \DataSpace$, where $\DataSpace$ is a vector space of data. 

An important component of indirect matching is the model for the data generation process, encoded by an operator (\emph{forward model})
\begin{equation}\label{eq:FwdOp}
  \ForwardOp \colon \ShapeSpace \to \DataSpace, 
\end{equation}  
which models how an element in $\ShapeSpace$ generates noise-free data in $\DataSpace$.
The indirect registration is performed by finding a group element $\gelem \in \LieGroup$ that parametrizes a deformation $\GroupAction \colon \LieGroup \times \ShapeSpace \to \ShapeSpace$ that minimizes a suitable energy. 
The energy necessarily contains a \emph{data fidelity}, i.e., 
a mapping $\LossFunc_{\DataSpace} \colon \DataSpace \times \DataSpace \to \Real$ that quantifies similarity in $\DataSpace$.

For many applications, especially if $\LieGroup$ is infinite-dimensional, the problem needs to be regularized to avoid excessive and unrealistic deformation, so typically a regularization $\mathcal{R} \colon \LieGroup \to \Real$ is added, resulting in the energy
\begin{equation}\label{eq:InDirectRegEnergyInGroup}
  \gelem \mapsto
\LossFunc_{\DataSpace}\Bigl((\ForwardOp \circ \GroupAction)\bigl(\gelem,\template\bigr),\target\Bigr) 
    +  \lambda \mathcal{R}(\gelem) \quad \text{for $\gelem \in \LieGroup$.}
\end{equation}
 Here, $\lambda > 0$ determines the amount of regularization.

A common choice of regularization, used in geodesic shape matching (also known as large deformation diffeomorphic metric mapping), see e.g., \cite{Bruveris2013}, is to use 
\begin{align}\label{eq:LDDMM}
    \mathcal{R}(\gelem) = \LieGroupDistance\bigl(\gelem,e \bigr)^2
    \quad\text{for $\gelem \in \LieGroup$.}
\end{align}
Here, $e \in \LieGroup$ is the identity element and
 $\LieGroupDistance \colon \LieGroup \times \LieGroup \to \Real$ is a distance on $\LieGroup$.

An issue with minimizing the energy in \cref{eq:InDirectRegEnergyInGroup} is that the evaluation of $\LieGroupDistance$  involves solving an optimization problem, so the minimization in \cref{eq:InDirectRegEnergyInGroup} is a coupled optimization problem over a Lie group. 
One can overcome this by choosing a distance that is a functional of a right-invariant Riemannian metric.
In this case, the minimization of the functional in \cref{eq:InDirectRegEnergyInGroup} can be replaced with minimizing over curves $\gcurve \colon [0,1] \to \LieGroup$ in the Lie group $\LieGroup$:
\begin{equation}\label{eq:InDirectRegEnergyInCurveOnGroup}
  \gcurve \mapsto
\LossFunc_{\DataSpace}\Bigl( (\ForwardOp\circ \GroupAction)\bigl(\gcurve(1),\template\bigr),\target\Bigr) 
    +  \lambda\int_0^1 \bigl\langle \dot\gcurve(t),\dot\gcurve(t) \bigr\rangle_{\gcurve(t)} \dint t  
    \quad\text{such that $\gcurve(0)=e$,}
\end{equation}
with $\langle \Cdot , \Cdot \rangle$ denoting a right-invariant metric on $\LieGroup$. 
The above, however, involves minimizing a functional over set of curves in the Lie group (or generating them from a curve of vector fields by solving the flow equation)  and therefore incurs additional computational cost. 

This paper considers indirect shape matching directly on the Lie group as in \cref{eq:InDirectRegEnergyInGroup} by using gradient flows.
We consider the problem of reconstructing protein backbone conformations from noisy cryo-SPA data, as in Jansson et al. \cite{Jansson2025}. 
Here, the group is a direct product of the special orthogonal group in $\Real^{3\times 3}$.

\section{Gradient flows for indirect shape matching}
\label{sec:gfs_ism}
In this section, we derive, in a general geometric setting, gradient flows for indirect shape matching problems. 
We begin by presenting the geometric background, following Balehowsky et al. \cite{BaKaMo2022}.
There are two central concepts. 
First, that of group actions and second, that of right-invariant Riemannian metrics. 

First, note that the Lie group $\LieGroup$ acts on itself from the right and the left, i.e., by right and left translation. 
We denote these actions by the mappings $\leftact[]\colon \LieGroup \times \LieGroup \to \LieGroup$ and $\rightact[]\colon \LieGroup \times \LieGroup \to \LieGroup$, respectively. 
Further, as $\LieGroup$ also is a manifold, it is possible to deduce the action on tangent spaces induced by the right and left translations. 
For right translations, this is simply given by the differential of $\rightact[\gelem]$ at $\gelemother$, i.e., $\gelem$ acts on tangent vectors $\xi \in T_\gelemother \LieGroup$ by $\xi \mapsto d(\rightact[\gelem])_\gelemother \xi$. 
As a shorthand, we write $T\rightact[\gelem] \xi$.  

For the group action $\Phi$, the \emph{orbit} of a point $\template \in \ShapeSpace$ is given by 
\begin{align*}
    \orb[\template] = \GroupAction(\LieGroup,\template) = \bigl\{  \GroupAction(\gelem,\template): \gelem \in \LieGroup \bigr\}. 
\end{align*}
The action $  \GroupAction$ of $\LieGroup$ on $\ShapeSpace$ gives rise to an \textit{infinitesimal action} of the Lie algebra $\LieAlgebra$. 
With a fixed $\template \in \ShapeSpace$, it is given by the differential of $\GroupAction$ at the identity, and is a mapping from $\LieAlgebra\times \ShapeSpace \to T \ShapeSpace$, the tangent bundle of $\ShapeSpace$.  
If a $\aelem \in \LieAlgebra$ is fixed, we can obtain a vector field $\aelem_\GroupAction\colon \ShapeSpace \to T\ShapeSpace$ that generates $\Phi$ in the sense that $\aelem_\GroupAction(w)=d\bigl(\Phi(\cdot,w)\bigr)_e \aelem$. 
Central to the following presentation is the concept of the \emph{momentum map} that arises by the left action of $\LieGroup$ on $\ShapeSpace$ lifted to the cotangent space $T^*\ShapeSpace$. 
This is a mapping $J\colon T^* \ShapeSpace \to \LieAlgebra^*$, i.e., it maps from the cotangent space of the shape space into the dual of the Lie algebra of $\LieGroup$. 
The momentum map is given by 
\begin{equation}\label{eq:mommap}
\pair[\LieAlgebra^*,\LieAlgebra]{J(\template,p)}{\aelem} = \pair[T^*_\template \ShapeSpace,T_\template \ShapeSpace]{p}{\aelem_\GroupAction(\template)}.
\end{equation}
Here, $p \in T^*_\template \ShapeSpace$ and $\pair[\Cdot,\Cdot]{\Cdot}{\Cdot}$ denotes the dual pairing between the spaces indicated in the subscript.

As our goal is to work with gradient flows, and gradients are Riemannian objects, we equip $\LieGroup$ with a metric. 
At the point $\gelem \in \LieGroup$, we denote it by $\inpro[g]{\Cdot}{\Cdot}\colon T_\gelem \LieGroup \times T_\gelem \LieGroup \to \Real$. 
As noted in \cite{BaKaMo2022}, if this metric is right-invariant, i.e., if $\inpro[e]{\aelem}{\aelemother} = \inpro[g]{T\rightact[\gelem] \aelem}{T\rightact[\gelem] \aelemother}$ for all $\gelem \in \LieGroup$ and $\aelem, \aelemother \in \LieAlgebra$, it is completely determined by a choice of inner product at the tangent space, and we write 
\begin{align*}
    \inpro[e]{\aelem}{\aelemother} = \pair[\mathfrak{g}^*, \mathfrak{g}]{\inertia\aelem}{\aelemother},
\end{align*}
where $\inertia\colon \LieAlgebra \to \LieAlgebra^*$ is the \emph{inertia operator} defining the inner product.
It is symmetric and positive-definite. 
The gradient of a mapping $E\colon \LieGroup \to \Real$ is now defined by
\begin{align}
    \label{eq:grad}
    \pair[T_\gelem^* \LieGroup,T_\gelem \LieGroup]{dE_\gelem}{TR_\gelem \aelem} = \inpro[\gelem]{\nabla E(\gelem)}{TR_\gelem \aelem}
\end{align}
for all $\aelem \in \LieAlgebra$. 
The following result, due to \cite[Proposition 2.3]{BaKaMo2022} relates $\nabla E$ to the geometric structure outlined above. 
\begin{theorem}
    \label{th:gradient_general}
    Assume that the Lie group $\LieGroup$ is equipped with a right-invariant metric and that it acts from the left on a shape space $\ShapeSpace$ by $\GroupAction\colon \LieGroup \times \ShapeSpace \to \ShapeSpace$.
    Let $J: T^* \ShapeSpace \to \LieAlgebra^*$ be the momentum map corresponding to this action.
    Let $f$ be a sufficiently smooth mapping $f \colon \ShapeSpace \to \Real$, and define for a fixed $\template \in \ShapeSpace$ the mapping $E(\gelem) = f\bigl(\Phi(g,\template)\bigr)$ from $\LieGroup$ to $\Real$. 
    Then, the gradient of $E$ is given by
    \begin{align*}
        \nabla E(\gelem) = T\rightact[\gelem] \inertia^{-1} J\Bigl(\Phi(\gelem,\template),df\bigl(\Phi(\gelem,\template)\bigr)\Bigr).
    \end{align*}
\end{theorem}
Thus, computing the gradient comes down to computing the differential of the mapping $f\colon \ShapeSpace \to \Real$ and inserting it into the momentum map.
This is relatively simple, but to avoid confusion, we shall describe the general procedure. 
In the following, let $M,N, O$ be general smooth manifolds. 
The differential of a function $f \colon M \to N$ is a mapping that takes a point $p \in M$ and sends it to a linear mapping from $T_p M$ to $T_{f(p)} N$. 
In the case when $N$ is the real space $\Real$, we see that the differential thus results in an element of $T_p^* M$. 
Generally, for a general vector $v_p \in T_p M$, where we understand vectors as derivations acting on smooth functions on the underlying manifold, the differential  is defined as $dF(p)[v_p](\varphi) = v_p(\varphi \circ F)$, where we see that $f \circ F$ is a smooth function on $M$, so it can be acted upon by $v_p$.   
The chain rule is simply, given a $g \colon N \to O$, given by 
\begin{align*}
    d(g \circ f)(p) = dg\bigl(f(p)\bigr) \circ df(p) \colon T_p M \to T_{g \circ f(p) } O,
\end{align*}
passing through the tangent space $T_{f(p)}N$.

Equipped with \cref{th:gradient_general} to compute gradients on $\LieGroup$, we can turn to the \emph{gradient flow}. 
Given a functional $E\colon \LieGroup \to \Real$, the gradient flow induced by $E$ is the differential equation 
\begin{align}
    \label{eq:general_form}
    \dot \gelem = -\nabla E(\gelem).
\end{align}
A classical application of gradient flows is to find minimizers of a functional. 
Indeed, if $\gelem(t)$ is the path in $\LieGroup$ induced by  \cref{eq:general_form}, we have by the chain rule that 
\begin{align}
    \label{eq:en_decrease}
    \frac{\mathrm d}{\mathrm d t} E(\gelem) = \pair[T_\gelem^* \LieGroup,T_\gelem \LieGroup]{dE_\gelem}{\dot \gelem} =  \inpro[\gelem]{\nabla E(\gelem)}{\dot \gelem} =  -\inpro[\gelem]{\nabla E(\gelem)}{\nabla E(\gelem)} \leq 0
\end{align}
Of course, this might converge to a minimum, a set of critical points, or indeed not converge at all. 
The exact convergence analysis depends heavily on the choice of $\LieGroup$ and $E$, as must be done on a case-by-case basis.

An important question is if there exists solutions at all to the equation in \cref{eq:general_form}. 
In the setting when $\LieGroup$ is finite dimensional, the answer is affirmative. 
By Balehowsky et al. \cite[Theorem 2.7]{BaKaMo2022}, we have that local Lipschitz continuity of the gradient of $E$ is sufficient to ensure a unique global solution $\gelem\colon[0,\infty) \to \LieGroup$. 

We are now in position to formulate gradient flows directly on the group that solves problems of the form in \cref{eq:InDirectRegEnergyInGroup}. 
We start with a template shape $\template \in \ShapeSpace$ and let it be acted upon by elements $\gelem$ of $\LieGroup$. 
Now, we would like to find the element of $\ShapeSpace$ that (in some suitable sense) is the best reconstruction of some observed data $\target$ in $\DataSpace$. 
To achieve this, we consider first the unregularized case. 
After selecting an appropriate data fidelity
$\LossFunc_{\DataSpace} \colon \DataSpace \times \DataSpace \to \DataSpace$ and a forward model $\ForwardOp\colon \ShapeSpace \to \Real$, we construct the mapping 
$\template \mapsto \LossFunc_{\DataSpace}\bigl(\ForwardOp(\template),\target\bigr)$
from $\DataSpace$ to $\Real$. 
This serves as the mapping $f$ in \cref{th:gradient_general}, and so $E\colon \LieGroup \to \Real$ is just $\LossFunc_{\DataSpace}\bigl(\ForwardOp(\GroupAction(g,\template)),\target\bigr)$. 
Thus, computing the gradient amounts to computing the differential of $\LossFunc_{\DataSpace}\bigl(\ForwardOp(\template),\target\bigr)$.
When we write down the differential to compute $\nabla E$, it is important to record which mappings are between which spaces. 
To simplify, we use the notation $\LossFunc_{\data}\bigl(\ForwardOp(\template)\bigr) = \LossFunc_\DataSpace\bigl(\ForwardOp(\template), \data \bigr)$.
Note that this function maps $\ShapeSpace$ to $\Real$. 
By applying the chain rule of differentials, we thus obtain that 
\begin{align*}
    d\LossFunc_{\data}\bigl(\ForwardOp(\template)\bigr) = d\LossFunc_\DataSpace\bigl(\ForwardOp(\template), \data \bigr) \circ d\ForwardOp(\template),
\end{align*}
where we see that $d\LossFunc_\DataSpace(\ForwardOp(\template), \data)$ maps $T_{\ForwardOp(\template)} \DataSpace \to \Real$ (with the differential taken with respect to the first argument, holding $\data$ fixed), and $d \ForwardOp(\template)$ maps $T_{\template} \ShapeSpace$ to $T_{\ForwardOp(\template)} \DataSpace$. 
Note that the composition should be understood as compositions of  mappings between tangent spaces. 

The gradient of $E$ is
\begin{align*}
        \nabla E(\gelem) = T\rightact[\gelem] A^{-1} J\Bigl(\Phi(\gelem,\template),d\LossFunc_\DataSpace\bigl(\ForwardOp\bigl(\GroupAction(\gelem,\template)\bigr),y\bigr) \circ d\ForwardOp\bigl(\GroupAction(\gelem,\template)\bigr)\Bigr).
\end{align*}
In practice, when $\ShapeSpace$ and $\DataSpace$ are Hilbert spaces, 
we can compute the differential by the following steps:
\begin{enumerate}
    \item Compute the forward pass $\dataother = \ForwardOp(\template)$
    \item Compute the gradient of $\LossFunc_\DataSpace$ with respect to its first argument, denoted by $\nabla_{\dataother} \LossFunc_\DataSpace$, and evaluate it at $\dataother$.
    \item Then, pull back through the forward model by its adjoint, i.e., compute $d\LossFunc_\DataSpace(\template) = d\ForwardOp(\template)^* [\nabla_{\dataother} \LossFunc_\DataSpace]$ where $d\ForwardOp(\template)^*:T^*_{\ForwardOp(\template)}\DataSpace \to T^*_{\template}\ShapeSpace$ is the operator satisfying
 \[\pair[T_\template^* \ShapeSpace, T_\template \ShapeSpace]{d\ForwardOp(\template)^* \alpha}{v} = \pair[T_{\ForwardOp(\template)}^* \DataSpace,T_{\ForwardOp(\template)} \DataSpace]{\alpha}{d\ForwardOp(\template)[v]}
 \] for all $\alpha \in T_{\ForwardOp(\template)}^* \DataSpace$ and $v \in T_{\template} \ShapeSpace$.
\end{enumerate}
Regularization can be included in several ways. 
While no additional regularization is needed to prove for instance, existence of minimizers if the group is compact, a regularization can be added by letting $\LieGroup$ act on $\ShapeSpace$ and some other space simultaneously so that the geometric setting is recovered with $\ShapeSpace' = \ShapeSpace \times \hat \ShapeSpace$. 
This can be used to avoid for instance self-intersection or to weakly enforce desired behaviors in the reconstruction.

\section{Shape gradient flows for protein conformation reconstruction}

In this section, we adapt the general theory of indirect shape matching gradient flows to the setting of single-particle cryo–EM. 
First, we describe the relevant shape space, data space,  forward model, and Lie group. 
Then, we derive and analyze the resulting flow obtained by adapting \cref{th:gradient_general}.

\subsection{Protein structures and shape setting}
\label{ssec:protein_shape}
Consider the set of 3D arrangements of the \ch{C_{$\alpha$}} atoms in the backbone of a fixed protein, denoted by  $\AtmModSpace$. Its elements are arrays  $\bbpos=(\bbpos_1,\dots,\bbpos_{\numPep})\in\Real^{3\times \numPep}$ where $\numPep$ is the number of residues in the protein and each $\bbpos_i \in \Real^3$ is the position of the $i$:th \ch{C_{$\alpha$}} atom.
From such a representation, one can compute the electrostatic potential generated by the \ch{C_{$\alpha$}} atoms in the backbone.
It is represented by a map $\AtmToMapOp \colon \AtmModSpace \to \MapSpace$ and it is potentially a sufficiently good approximation of the 3D map for the entire protein.
We make the further simplifying assumption that the interatomic distances between the \ch{C_{$\alpha$}} atoms are constant over deformations. 
This allows for representing the backbone structure via the relative positions of the \ch{C_{$\alpha$}} atoms.
To define the shape space, we first take all relative positions of the  \ch{C_{$\alpha$}} atoms in the backbone, that is, an array in $\Real^{3\times (\numPep-1)}$:
\[
  (\bbpos_2-\bbpos_1,\dots,\bbpos_{\numPep}-\bbpos_{\numPep-1})
  \in\Real^{3\times (\numPep-1)}
  \quad\text{with}\quad
  (\bbpos_1,\bbpos_2, \dots,\bbpos_{\numPep})
  \in\AtmModSpace.
\]
This could be the shape space. However, it is important to have a space bijective to $\AtmModSpace$, so we supplement the relative positions with the position of the first \ch{C_{$\alpha$}} atom, i.e., our shape space $\ShapeSpace$ consists of $\numPep$-arrays of the form 
\[
  \relbbpos=(\bbpos_1,\bbpos_2-\bbpos_1,\dots,\bbpos_{\numPep}-\bbpos_{\numPep-1})
  \in\Real^{3\times \numPep}
  \quad\text{with}\quad
  (\bbpos_1,\bbpos_2, \dots,\bbpos_{\numPep})
  \in\AtmModSpace.
\]
As noted,  $\AtmModSpace$ and $\ShapeSpace$ are related via the bijection $\opStyle{M} \colon \ShapeSpace \to \AtmModSpace$ given by
\[
    \opStyle{M}(\relbbpos):=\Bigl(\relbbpos_1, \relbbpos_1+\relbbpos_2, \relbbpos_1+\relbbpos_2+\relbbpos_3,\dots,\sum_{j=1}^\numPep \relbbpos_j \Bigr)
\]
whose inverse $\opStyle{M}^{-1} \colon \AtmModSpace \to \ShapeSpace$  is
\[
    \opStyle{M}^{-1}(\bbpos)=(\bbpos_1, \bbpos_2-\bbpos_1,\dots,\bbpos_{\numPep}-\bbpos_{\numPep-1}).
\]

To define the Lie group that acts on elements of the shape space defined above, let first $\relbbpos, \relbbposother \in \ShapeSpace$ be two different backbone structures of the protein. 
Since these have the same interatomic distances, there are rotation matrices $\rotmat_1,\dots, \rotmat_{\numPep}\in\SO(3)$ such that $\relbbposother_i=\rotmat_i \relbbpos_i$. 
The deformation of the backbone $\relbbpos$ into $\relbbposother$ can then be represented by acting with the Lie group $\LieGroup := \bigl(\SO(3)\bigr)^\numPep$ on $\ShapeSpace$ via component-wise matrix-vector multiplication:
\[  
\GroupAction \colon \LieGroup \times \ShapeSpace \to \ShapeSpace
\quad\text{where}\quad
\GroupAction(\gelem,\relbbpos) := 
\bigl( \rotmat_1\relbbpos_1, \ldots, \rotmat_{\numPep}\relbbpos_{\numPep} \bigr),
\quad\text{with $\gelem=(\rotmat_1,\ldots,\rotmat_{\numPep} ) \in \LieGroup$.}
\]
This means that we act with the Lie group $\bigl(\SO(3)\bigr)^\numPep$ on the relative atom positions for the \ch{C_{$\alpha$}} atoms in the backbone. 

Protein conformations are only indirectly observable by \ac{TEM} imaging, so we decide on a forward operator to model how a deformable object in  $\ShapeSpace$ maps to an observable 2D \ac{TEM} image. 
We view 2D \ac{TEM} images as digitized functions in $\LpSpace^2(\Real^2)$, the space of square-integrable functions on $\Real^2$.
In the following, we have $\numImgs$ \ac{TEM} images $\data_i \in \LpSpace^2(\Real^2)$ for $i=1,\ldots, \numImgs$. 
Each \ac{TEM} image $\data_i$ corresponds to a backbone structure $\relbbpos_i \in \ShapeSpace$. 
The forward map models the data generation process. 
First, $\relbbpos_i$ is mapped to the corresponding 3D point cloud representation in $\AtmModSpace$ via $\opStyle{M} \colon \ShapeSpace\to \AtmModSpace$. 
Next, we create an approximate 3D map that is the input for the model of the \ac{TEM} image formation.   
To achieve this, we replace each atom in the backbone with a 3D Gaussian centered at the atom.
Formally, this is a map $\AtmToMapOp \colon \AtmModSpace\to \MapSpace$ where $\MapSpace = \LpSpace^2(\Real^3)$. 

Then, we project onto the detector plane via the parallel beam ray transform $\opStyle{P}\colon \MapSpace \to \LpSpace^2(\Real^2)$.
Finally, we generate a 2D \ac{TEM} image from the projected 3D map by an operator $\opStyle{C}\colon\LpSpace^2(\Real^2)\to\LpSpace^2(\Real^2)$ that models the microscope optics. 
We define $\opStyle{C}$ as a convolution with a chosen point spread function (PSF) $h$.
The Fourier transform of the PSF is called a contrast transfer function (CTF).
Following Erickson and Klug \cite{Erickson1970} (see also Fanelli and Öktem \cite{Fanelli2008} and Yang et al. \cite{Yang2009}), we use the following definition of the CTF
\begin{equation}\label{eq:ctf}
     \hat h(\xi) = - \Bigl(\sqrt{1-\alpha^2} \sin\bigl(\gamma(\xi)\bigr) + \alpha \cos\bigl(\gamma(\xi)\bigr)\Bigr)e^{-\frac{B}{2} |\xi|^2}
\end{equation}
where 
\begin{align*}
\gamma(\xi) = 2\pi \left(-\frac{C_s\lambda^3 |\xi|^4}{4} + \frac{\Delta z \lambda |\xi|^2}{2}\right).
\end{align*}
Here, $\alpha \in [0,1]$ is the amplitude contrast ratio, $B$ is the experimental $B$-factor, $\Delta z $ the defocus, $C_s $ is the spherical aberration, and $\lambda$ is the electron wave length. 
The CTF gives rise to a corresponding map $\opStyle{C}\colon\LpSpace^2(\Real^2) \to \LpSpace^2(\Real^2)$ constructed by the following steps:
\begin{enumerate}
    \item Take the Fourier transform $\hat f $ of the input $f \in \DataSpace$. 
    \item Multiply $\hat f$ with $\hat h$ pointwise. 
    \item Compute the inverse Fourier transform of the product to obtain $\opStyle{C}f$.
\end{enumerate}
Note that $\opStyle{C}$ is a convolution of the input with the point spread function, and is thus a linear mapping. 
The map $\opStyle{C}$ is self-adjoint.
To see this, we only need to observe that multiplication with $\hat h$ is self-adjoint, as the Fourier transform is an isometry of $\LpSpace^2(\Real^2)$ and thus has adjoint equal to its inverse.
Multiplication with $\hat h$ is self-adjoint since $\hat h$ is real:
\begin{equation*}
    (u,\hat hv)=\int_{\Real^2}u(\xi)\overline{\bigl(\hat h(\xi)v(\xi)\bigr)}d\xi=\int_{\Real^2}u(\xi)\hat h(\xi)\overline{v(\xi)}d\xi=(\hat hu,v).
\end{equation*}
See \cite{Fanelli2008, Oktem:2015aa} for example for more details on the image formation model.

The forward operator $\ForwardOp \colon \ShapeSpace \to \DataSpace$ is now given by (see also \cref{fig:forward_model} for an illustration) 
\begin{equation}\label{eq:TEMtotFwdOp}
\ForwardOp \coloneqq \opStyle{C}\circ\opStyle{P}\circ\AtmToMapOp\circ\opStyle{M}.
\end{equation}
When $\opStyle{P}$ is the parallel beam ray transform, then one can evaluate the first three maps in \cref{eq:TEMtotFwdOp} in a computationally feasible manner as follows: Perform a geometric projection of the input 3D point cloud along the \ac{TEM} optical axis onto the 2D  \ac{TEM} detector plane. Then, apply the 2D analogue of $\AtmToMapOp$ to the 2D point cloud. \cite{Oktem:2015aa}
\begin{figure}
    \centering
    \includegraphics[width=0.7\linewidth]{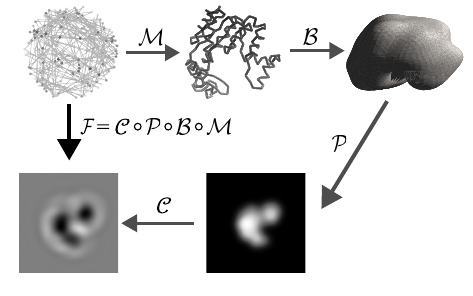}
    \caption{Illustration of the forward model. First, the relative coordinates are mapped to absolute coordinates. Then, a Gaussian blob is added to each bead, after which a projection is applied. Finally, one applies the CTF.}
    \label{fig:forward_model}
\end{figure}

The data consists of a set of observed images $\data_1,\ldots,\data_\numImgs$. 
We first make the simplifying assumptions that
\begin{enumerate}
    \item The \emph{orientation} is known in each image.
    \item The data contains only one conformation, i.e., is conformationally homogeneous.
\end{enumerate} allowing us to view the full data space as the direct sum 
\begin{align*}
    \DataSpace = \bigoplus_{j=1}^\numImgs \LpSpace^2(\Real^2),
\end{align*}
In the following, we denote by $\DataSpace_j$ the $j$th component of $\DataSpace$ and by 
 $\ForwardOp_j\colon \DataSpace_j \to \Real$ be the $j$th component of  $\ForwardOp\colon \ShapeSpace \to \DataSpace$. It is given by 
\begin{equation}\label{eq:TEMtotFwdOp_pw}
\ForwardOp_j \coloneqq \opStyle{C}\circ\opStyle{P}_j\circ\AtmToMapOp\circ\opStyle{M},
\end{equation}
where $\opStyle{P}_j$ is the projection along the beam direction generating the image.
As discussed above, this is in practice obtained by
the equivalent mapping
\begin{equation}\label{eq:TEMtotFwdOp_simp}
\ForwardOp_j \coloneqq \opStyle{C}\circ \AtmToMapOpPlane \circ\opStyle{\pi}_j\circ\opStyle{M},
\end{equation}
where 
\begin{equation}\label{eq:gauss}
\AtmToMapOpPlane(\amodplane)(\Cdot) =
    \sum_{i=1}^{\numPep} s_i \gauss_{\amodplane_i, \sigma_i}(\Cdot).
\end{equation}
Here, $\gauss_{\mu,\sigma}\colon \Real^2 \to \Real$ is the Gaussian probability density function 
\begin{align*}
    \gauss_{\mu,\sigma}(x) = \frac{1}{2\pi\sigma^2} \exp\left(-\|x-\mu\|^2/(2\sigma^2)\right),
\end{align*}
$s_1,\dots,s_{\numPep}$ are positive weights 
and $\sigma_1^2,\dots,\sigma_{\numPep}^2$ are the variances of the respective Gaussians. 
The mapping $\pi_j \colon \AtmModSpace \to \Real^2$ is obtained by applying a rotation $\rotmat_i$ to $\opStyle{M}(\template)$, and then removing the final component, i.e., 
\begin{align*}
    \pi_j(x) = \begin{pmatrix} (\rotmat_j x)_1, (\rotmat_j x)_2 \end{pmatrix}.
\end{align*}

\subsection{Derivation of the gradient flow
}
To apply the gradient flow approach of \cref{sec:gfs_ism} to the present problem, we first need to formulate it as a shape matching problem.
Much of the groundwork was already laid in \cref{ssec:protein_shape}. 
It remains to specify an appropriate matching functional.
We first focus on the unregularized case, so
the goal is thus to find an element $\gelem \in \SO(3)^\numPep$ that minimizes the functional
\begin{equation}
    \label{eq:so3_prob_unreg}
   \gelem \mapsto \LossFunc_\data(\gelem) \coloneqq \sum_{j=1}^\numImgs \LossFunc_{\DataSpace_j}\Bigl(\ForwardOp_j\bigl(\GroupAction(\gelem, \template)\bigr), \data_j\Bigr)
\end{equation}
over $\SO(3)^\numPep$. 
Here, $\LossFunc_{\DataSpace_j}$ denotes the $j$th component of the loss. 
We now present the framework of \cref{sec:gfs_ism} in the protein reconstruction setting of \cref{ssec:protein_shape}. 

The infinitesimal generator of the action is simply $\hat \aelem \template$, where $\hat \aelem \in \SOLieAlgebra(3)^\numPep$ consists of $\numPep$ skew-symmetric matrices. 
Component-wise, we can apply the hat-map that sends a skew-symmetric matrix $\hat \aelem_i$ to a vector $\aelem_i$ in $\Real^3$ and instead view the Lie algebra element as a list of such vectors. The $i$th component of the infinitesimal action $\hat \aelem \template$ is then just $\aelem_i \times \template_i$, where $\times$ is the $\Real^3$ cross product. 
The momentum map likewise separates into components, where its $i$th component acts on the $i$th component of $(\template,p) \in T^*\ShapeSpace$ by 
\begin{equation}
    \label{eq:so3_mom}
    J(\template,p)_i = \template_i \times p_i. 
\end{equation}
Further, the tangent lifted action $T\rightact[\gelem]$ just becomes matrix multiplication from the right. 
We equip the algebra with an inner product defined by an inertia operator, that is, a mapping $\inertia \colon \SOLieAlgebra(3)^\numPep \to \bigl(\SOLieAlgebra(3)^\numPep\bigr)^* $.
We take $\inertia$ as a list of symmetric and positive definite  operators each acting on only one component, i.e., $\inertia = (\inertia_1, \ldots \inertia_\numPep)$. 
The inner product is then given by $\langle \aelem, \aelemother \rangle_\inertia = \traceop(\inertia \aelem, \aelemother)$. We denote the corresponding norm by $\|\cdot \|_\inertia$. 

The metric on $\SO(3)^\numPep$ is thus defined by 
\begin{align}
    \label{eq:metric_so3}
    \inpro[g]{\dot \gelem_1}{\dot \gelem_2} = \traceop(\inertia(\dot \gelem_1 \gelem^{-1}), \dot \gelem_2 \gelem^{-1}),
\end{align}
where $\dot \gelem_1,\dot \gelem_2 \in T_\gelem \SO(3)^\numPep$. 
The corresponding norm is given by 
\begin{align*}
    \label{eq:norm_so3}
    \|\dot \gelem\|_\gelem^2 =  \traceop(\inertia(\dot \gelem \gelem^{-1}), \dot \gelem \gelem^{-1}) = \|\dot \gelem \gelem^{-1}\|^2_\inertia.
\end{align*}

Then, the gradient flow is given by the following corollary to \cref{th:gradient_general}
\begin{corollary}\label{corr:so3grad}
The flow equation 
    \begin{equation}
    \label{eq:gf_so3}
    \dot \gelem = -\eta \gelem
    \end{equation}
holds in which the $i$th component of $\eta$ is given by 
\[
    \inertia_i\eta_i=\GroupAction(\gelem, \template)_i\times\biggl(\sum_j\Bigl(d \ForwardOp_j\bigl(\GroupAction(\gelem, \template)\bigr)\Bigr)^*   \nabla_{\dataother_j} \LossFunc_{\DataSpace_j}(\dataother_j, \data_j) \Big|_{\dataother_j = \ForwardOp_j(\GroupAction(\gelem, \template))} \biggr)_i
\]
where
\begin{align*}
   \bigl(d\mathcal F_j(v)\bigr)^*f=\mathcal M^* \, \pi_j^* \,d\AtmToMapOpPlane\bigl(\mathcal M(v)\bigr)^*\,\mathcal Cf,
\end{align*}
with 
  \begin{align*}
          &d\AtmToMapOpPlane(\amodplane)^*\data = \begin{pmatrix} \data * s_1\nabla \gauss_{0,\sigma_1}(\amodplane_1), \ldots, \data* s_\numPep\nabla \gauss_{0,\sigma_\numPep}(\amodplane_\numPep)  \end{pmatrix} \in \Real^{2 \times \numPep },\\
          & \pi_j^* \begin{pmatrix} \amodplane_1, \ldots, \amodplane_\numPep \end{pmatrix} = \begin{pmatrix} \rotmat_j^\intercal \pi^* \amodplane_1, \ldots \rotmat_j^\intercal \pi^* \amodplane_\numPep \end{pmatrix}, \quad \text{with } \pi^* = \begin{bmatrix} 1 & 0 \\ 0 & 1 \\ 0 & 0 \end{bmatrix}\\
          & (\opStyle{M}^* \amod)_i = \sum_{k = i}^\numPep \amod_k.
      \end{align*}
\end{corollary}
\begin{proof}
    The first step is to compute the differential of the functional in \cref{eq:so3_prob_unreg}
    as a function from $\ShapeSpace$ to $\Real$. 
    It coincides with the gradient of the functional in \cref{eq:so3_prob_unreg}, which is the element of $\DataSpace$ with components
    \begin{align*}
       \bigl(\nabla_{\dataother} \LossFunc_{\DataSpace}\bigr)_j =   \nabla_{\dataother_j} \LossFunc_{\DataSpace_j}(\dataother_j, \data_j), 
       \qquad j = 1,\dots,\numImgs,
    \end{align*}
    evaluated at $\dataother_j = \ForwardOp_j(\GroupAction(\gelem,\template))$.
    Then, it remains to determine the adjoint of the differential of the forward mapping. 
    First, recall that mapping $\opStyle{C}$ is self-adjoint. 
      Then, by \cite{Jansson2025}, the adjoint of the differential of $\AtmToMapOpPlane$ is  
      \begin{align*}
          d\AtmToMapOpPlane(\amodplane)^*\data = \begin{pmatrix} \data * s_1\nabla \gauss_{0,\sigma_1}(\amodplane_1), \ldots, \data* s_\numPep\nabla \gauss_{0,\sigma_\numPep}(\amodplane_\numPep)  \end{pmatrix} \in \Real^{2 \times \numPep }. 
      \end{align*}
     
      The projection $\pi_i$ is  computed in two steps. 
      First, a point cloud $\amod$ is rotated by $\rotmat_j$. Then, it is projected by removing the $z$-coordinate. 
      Thus, the adjoint of $\pi_j$ is given by 
      \begin{align*}
          \pi_j^* \begin{pmatrix} \amodplane_1, \ldots, \amodplane_\numPep \end{pmatrix} = \begin{pmatrix} \rotmat_j^\intercal \pi^* \amodplane_1, \ldots \rotmat_j^\intercal \pi^* \amodplane_\numPep \end{pmatrix}, \quad \text{where } \pi^* = \begin{bmatrix} 1 & 0 \\ 0 & 1 \\ 0 & 0 \end{bmatrix}.
      \end{align*}
       Finally, the $i$th coordinate of adjoint of the relative-to-absolute coordinate map $\opStyle{M}$ applied to $\amod = \begin{pmatrix}\amod_1, \ldots, \amod_\numPep\end{pmatrix} \in \AtmModSpace = \Real^{3 \times \numPep}$ is given by 
       \begin{align*}
           (\opStyle{M}^* \amod)_i = \sum_{k = i}^\numPep \amod_k.
       \end{align*}
       Inserting the above into the momentum map in \cref{eq:so3_mom} yields the gradient flow.
\end{proof}
We remark that the geometric setting extends naturally to the multi-chain setting, i.e., where the protein does not consist of a single chain, but rather of $L$ chains, where we denote the number of residues in chain $i$ by $\numPep_i$. 
In this setting, the shape space of relative coordinates becomes $V_{\mathrm{multi}} = \Real^{3 \times \numPep_1} \times \ldots \times \Real^{3 \times \numPep_L}$, and it is acted upon by $\SO(3)^\numPep_1 \times \ldots \times \SO(3)^\numPep_L$. 
The difference lies in the forward map, in which the relative-to-absolute coordinate map is replaced by a chain-wise mapping.
Once the absolute positions are known, we can proceed as before and apply the same forward model as in the single-chain case. 

We now consider some theoretical properites of the gradient flow in \cref{eq:gf_so3} and the minimization problem in \cref{eq:so3_prob_unreg}. 
First, note that \cref{eq:gf_so3} admits global existence. \begin{proposition}\label{prop:so3_wellp}
    The gradient flow in \cref{eq:gf_so3} admits a unique global solution $\gelem\colon [0,\infty) \to \LieGroup$.
\end{proposition}
\Cref{prop:so3_wellp} is a direct consequence of \cite[Theorem 2.7]{BaKaMo2022}, which holds since all components of \cref{eq:gf_so3} are smooth. 
Further, by virtue of the compactness of $\SO(3)^\numPep$, there exists a minimizer, i.e., for each $\data = (\data_1, \ldots, \data_\numImgs)$, there is a minimizer $\gelem_{\data}$. 

We now consider stability, convergence in the sense of vanishing observational noise (cf. \cite[Propositions~2.2 and 2.3]{Lang2019}), and convergence of the flow to a set of critical points.
We put a standing assumption on the data fidelity. 
\begin{assumption}\label{ass:BasicAssumption}
    The data fidelity $\LossFunc_\DataSpace \colon \DataSpace \times \DataSpace \to \Real$ is assumed to satisfy:
    \begin{enumerate}
        \item[A1] The component-wise losses $\LossFunc_{\DataSpace_j}$ are jointly continuous in both arguments.
        \item[A2] For fixed $\data \in \DataSpace$, $\gelem \mapsto \LossFunc_\data(\gelem)$ is $C^2$ on $\operatorname{SO}(3)\numPep$.
        \item[A3] There is an equivalence relation $\sim$ on $\DataSpace$ such that $\LossFunc_\DataSpace(\dataother,\data) = 0 \iff \dataother \sim \data$.
    \end{enumerate}
\end{assumption}
The main proof vehicle is the compactness of the underlying group. 
\begin{proposition}[Stability]
    Assume $\data^n \to \data \in \DataSpace$ and let $(\gelem^n)_{n=1}^\infty$ be the sequence where $\gelem^n$ is a minimizer of $\LossFunc_{\data^n}$.  Then, there is a subsequence $(\gelem^{n_k})_{k=1}^\infty$ converging to some minimizer $\gelem^* \in \argmin_{\gelem} \LossFunc_\data(\gelem)$.
\end{proposition}
\begin{proof}
    By the compactness of $\LieGroup$, we extract a subsequence $\gelem^{n_k}$ converging to some $\gelem^* \in \LieGroup$.
    For any $\tilde{\gelem} \in \SO(3)^\numPep$, it holds that \begin{align}
        \label{eq:stab_intermediate}
        \LossFunc_{\data^{n_k}}(\gelem^{n_k}) \leq  \LossFunc_{\data^{n_k}}(\tilde{\gelem}).
    \end{align}
    By the smoothness of the forward map, the group action and joint continuity of the fidelity, both sides of \cref{eq:stab_intermediate} converge, and passing to the limit yields that
    \begin{align*}
        \LossFunc_\data(\gelem^*) \leq  \LossFunc_\data(\tilde{\gelem}).
    \end{align*}
    The result directly follows from the above as $\tilde{\gelem}$ was arbitrary. 
\end{proof}
We now prove convergence as observational noise vanishes. 
\begin{proposition}[Convergence]\label{prop:conv_vann}
    Let $\data \in \DataSpace$ be given and suppose that there is $\hat{\gelem}$ such that $\ForwardOp\bigl(\GroupAction(\hat{\gelem},\template)\bigr) \sim \data$.
    Next, let $(\data^n)_{n=1}^\infty \to \data$ where $\gelem^n$ is a minimizer of $\LossFunc_{\data^n}$. 
    Then, $(\gelem^n)_{n=1}^\infty$ has a subsequence converging to some $\gelem^*$ with $\ForwardOp(\GroupAction(\gelem^*, \template)) \sim \data$. In other words, the reconstruction reproduces the data up to equivalence relation $\sim$ in \cref{ass:BasicAssumption}.
\end{proposition}
\begin{proof}
    By compactness, we extract a subsequence $\gelem^{n_k}$ converging to some $\gelem^* \in \SO(3)^\numPep$.
    By the definition of the sequence $(\gelem^n)_{n=1}^\infty$, comparing with $\hat{\gelem}$ gives that
    \begin{align*}
        0 \leq \LossFunc_{\data^{n_k}}(\gelem^{n_k}) \leq \LossFunc_{\data^{n_k}}(\hat \gelem) \to \LossFunc_\data(\hat\gelem) = 0,
    \end{align*}
    where the limit is by joint continuity and that $\data^{n} \to \data$ and the final equality holds since $\ForwardOp(\GroupAction(\hat{\gelem},\template)) \sim \data$. 
   By joint continuity, we have that 
   $\LossFunc_{\data^{n_k}}(\gelem^{n_k}) \to \LossFunc_{\data}(\gelem^*)$ so that $\LossFunc_\data(\gelem^*) = 0$, so $\ForwardOp(\GroupAction(\gelem^*,\template)) \sim \data$ by the definition of $\sim$.
\end{proof}
We remark that the equivalence $\sim$ depends on the choice of data fidelity. 
For instance, if one chooses to work with $\LpSpace^2$-type norms, $\sim$ is equality. 

Finally we consider the convergence of the gradient flow itself. 
We prove that every limit point of the flow is a critical point that is not a local maximum. 
\begin{proposition}[Convergence of gradient flow]
    Let $\LossFunc_\data$ as defined in \cref{eq:so3_prob_unreg} be $C^2$ and let $\gelem(t)$ solve the gradient flow 
    \begin{align}
        \dot{\gelem} = -\nabla \LossFunc_\data(\gelem).
    \end{align}
    Then, the following holds:
    \begin{enumerate}
        \item $\LossFunc_\data(\gelem(t))$ decreases monotonically to a limit $\LossFunc_\infty$.
    \item$\bigl\|\nabla\LossFunc(\gelem(t))\bigr\|_{\gelem(t)} \to 0$ as $t \to \infty$.
        \item Every limit point of the flow as $t \to \infty$ is a critical point which is not a maximum.
    \end{enumerate}
\end{proposition}
\begin{proof}
    By \cref{eq:en_decrease}, the energy decreases along the flow. 
    Since $\LossFunc_\data\colon \SO(3)^\numPep \to \Real$ is a continuous function on a compact set, it is bounded from below and so $\LossFunc_\data(\gelem(t))$ converges monotonically to some $\LossFunc_\infty$. 

    Consider now the right-trivialized gradient from \cref{corr:so3grad}, i.e., the quantity $\eta(t)$. 
    It is an element of the Lie algebra $\SOLieAlgebra(3)^\numPep$, and 
    \begin{align*}       \left\|\nabla\LossFunc_\data(\gelem(t))\right\|_{\gelem(t)} = \left\|\eta(t)\right\|_{\inertia}.
    \end{align*}
    Further, it holds that 
    \begin{align}
        \label{eq:squared_int}
        \int_0^\infty \|\eta(t)\|^2_\inertia \mathrm{d}t=  -\int_0^\infty \frac{\mathrm{d}}{\mathrm{d}t} \LossFunc_\data(\gelem(t))\mathrm{d}t = \LossFunc_\data(\gelem(0))-\LossFunc_\infty < \infty. 
    \end{align}
    Set $F\colon \SO(3)^\numPep \to \Real$ to be the function 
    $F(\gelem) = \|\nabla\LossFunc_\data(\gelem)\|^2_\gelem$.
    Since $\LossFunc_\data$ is in $C^2$, $F$ is in $C^1$.
    Along the flow, we have that $F(\gelem(t)) = \|\eta(t)\|^2_\inertia$, and
    \begin{align*}
        \frac{\mathrm{d}}{\mathrm{d}t} F(\gelem(t)) = -\inpro[\gelem(t)]{\nabla F(\gelem(t))}{\nabla\LossFunc_\data(\gelem(t))}.
    \end{align*}
    Both $\nabla F$ and $\nabla\LossFunc_\data$ are continuous on a compact set, so they are bounded, and therefore, $\frac{\mathrm{d}}{\mathrm{d}t} F(\gelem(t))$ is bounded, meaning in turn that the mapping 
    $t \mapsto F(\gelem(t))$ is uniformly continuous, so taken together with \cref{eq:squared_int}  Barbalat's lemma yields that $F(\gelem(t)) \to 0$, and the second statement of the proposition follows. 
    
    Since $\SO(3)^\numPep$ is compact, the flow has at least one accumulation point. 
    That is, there is a sequence of times $(t_n)_{n=1}^\infty$ with $\lim_{n\to \infty} t_n = \infty $ and a point $\gelem^* \in \SO(3)^\numPep$ such that $\lim_{n \to \infty} \gelem(t_n) = \gelem^*$.
    By continuity, we have that 
    $F(\gelem^*) = \lim_{n \to \infty} F(\gelem(t_n)) = 0$, meaning that any such $\gelem^*$ also is a critical point. 
    Moreover, $\LossFunc_\data(\gelem^*) = \LossFunc_\infty$, and as the flow decreases, $\gelem^*$ cannot be a local maximum. 
\end{proof}
\section{Numerical illustrations}
In this section, we present some numerical illustrations.
Throughout the experiments we base the agreement between a predicted image and an observation with a squared normalized cross-correlation (NCC) fidelity.
For a predicted image $\dataother$ and an observed image $\data$, let
\begin{align*}
    \mathrm{NCC}(\dataother,\data) = \frac{\langle \dataother, \data\rangle^2}{\lVert \dataother \rVert^2\,\lVert \data \rVert^2}
    \in [0,1],
\end{align*}
with the inner product and norms taken over the image pixels.
The data-fidelity functional over the $\numImgs$ observations is
\begin{align*}
     \LossFunc_{\DataSpace} = \sum_{j=1}^\numImgs \bigl(1 - \mathrm{NCC}(\dataother_j, \data_j)\bigr).
\end{align*}
Each term vanishes exactly when $\dataother_j$ is a nonzero scalar multiple of $\data_j$, so the fidelity is invariant to the per-image contrast scale and sign. 
Note that in this case, the $\sim$ equivalence in \cref{prop:conv_vann} is equality up to linear scalings, i.e., $\dataother\sim \data \iff \dataother = c \data$ for some non-zero scalar $c$. 
Further, we assume that no component of $\data$, nor $\ForwardOp(\GroupAction(\gelem,\template))$ for any $\gelem \in \SO(3)^\numPep$, equals $0$ , so that the squared-NCC fidelity is well-defined and jointly continuous.

Its gradient with respect to the predicted images is the element of $\DataSpace$ with components
\begin{align*}
     \bigl(\nabla_{\dataother} \LossFunc_{\DataSpace}\bigr)_j = -\frac{2\langle \dataother_j,\data_j\rangle} {\lVert \dataother_j\rVert^2\,\lVert \data_j\rVert^2} \left(\data_j - \frac{\langle \dataother_j,\data_j\rangle} {\lVert \dataother_j\rVert^2}\,\dataother_j\right),
    \qquad j = 1,\dots,\numImgs,
\end{align*}
which enters in \cref{eq:gf_so3} to obtain the gradient flow that is used in the numerical illustration. 
\footnote{The source code is available via  GitHub at \url{https://github.com/jansson-erik/cryo_shape_gradientflow}.}

\subsection{Data generation}\label{sec:datagen}
In all cases, we use synthetic data. 
For a set of absolute all-atom coordinates, i.e., not the coarse-grained \ch{C_{$\alpha$}} model, we first apply one global rotation $\rotmat \in \SO(3)$ representing the orientation of the protein. 
Denote the rotated coordinates of atom $i$ by $\bbpos_i' = \rotmat \bbpos_i$.
First, each atom is projected to the 2D image plane by removing the final coordinate. Denote the projected coordinate by $ \bbpos_j'' \in \mathbb{R}^2$.
Then, the projected potential is assembled in Fourier space, by summing over atom type $e$,
\begin{align*}
    \hat I(\xi) = \sum_e f_e\bigl(|\xi|\bigr) \hat{\delta}_e(\xi),
    \quad\text{with}\quad
    f_e\bigl(|\xi|\bigr) = \sum_{\alpha = 1}^4 a_\alpha^e \exp(-b_\alpha^e |\xi|^2),
\end{align*}
where $e_j$ is the element of atom $j$, $\xi \in \Real^2$ is the frequency, $\hat{\delta}_e(\xi)=\sum_{j:e_j=e} \exp\bigl(-(2\pi i)  \xi \cdot \bbpos_j''\bigr)$ and $f_e$ are the Doyle--Turner $4-$Gaussian electron scattering factors, with coefficient values taken from Doyle and Turner \cite{Doyle:1968}.

In real space this is a sum of four 2D Gaussians of different widths centred at each projected atom position.
We then apply the CTF in \cref{eq:ctf} with $\alpha = 1$, $\Delta z = 15{,}000\,\text{Å}$, $C_s = 2.7\times10^7\,\text{Å}$, $\lambda = 0.01969\,\text{Å}$ and a $B$-factor of $B_{\mathrm{signal}}=100\,\text{Å}^2$. 
The images are $s \times s$ with $s = 128$ pixels over a size of 
$L = 120\,\text{\AA}$, giving a pixel size of $L/s \approx 0.94\,\text{\AA}$. 

To add noise to the images, we add \emph{correlated} Gaussian noise $\varepsilon$, sampled by drawing white Gaussian noise and filtering it in Fourier space by a Gaussian amplitude envelope with $B_{\mathrm{noise}} = 30\,\text{Å}^2$.
After inverse FFT is applied to $\varepsilon$, we scale it so that the SNR reaches a desired level, by setting 
\begin{align*}
    \sigma_{\mathrm{noise}} = \frac{\sigma_{\mathrm{image}}}{\sqrt{SNR}},
\end{align*}
where $\sigma_{\mathrm{image}}$ is the standard deviation of the clean image intensities $I$, taken over all pixels of all $\numImgs$ images. 
The final image is given by 
\begin{align*}
    \data = I + \varepsilon. 
\end{align*}
The image formation is illustrated in \cref{fig:imgen}. 
Because the noise is not white, the frequency-wise signal-to-noise ratio varies: it is higher at low frequencies and lower at high frequencies (see \cref{fig:imgen}).
\begin{figure}
    \centering
    \includegraphics[width = \linewidth]{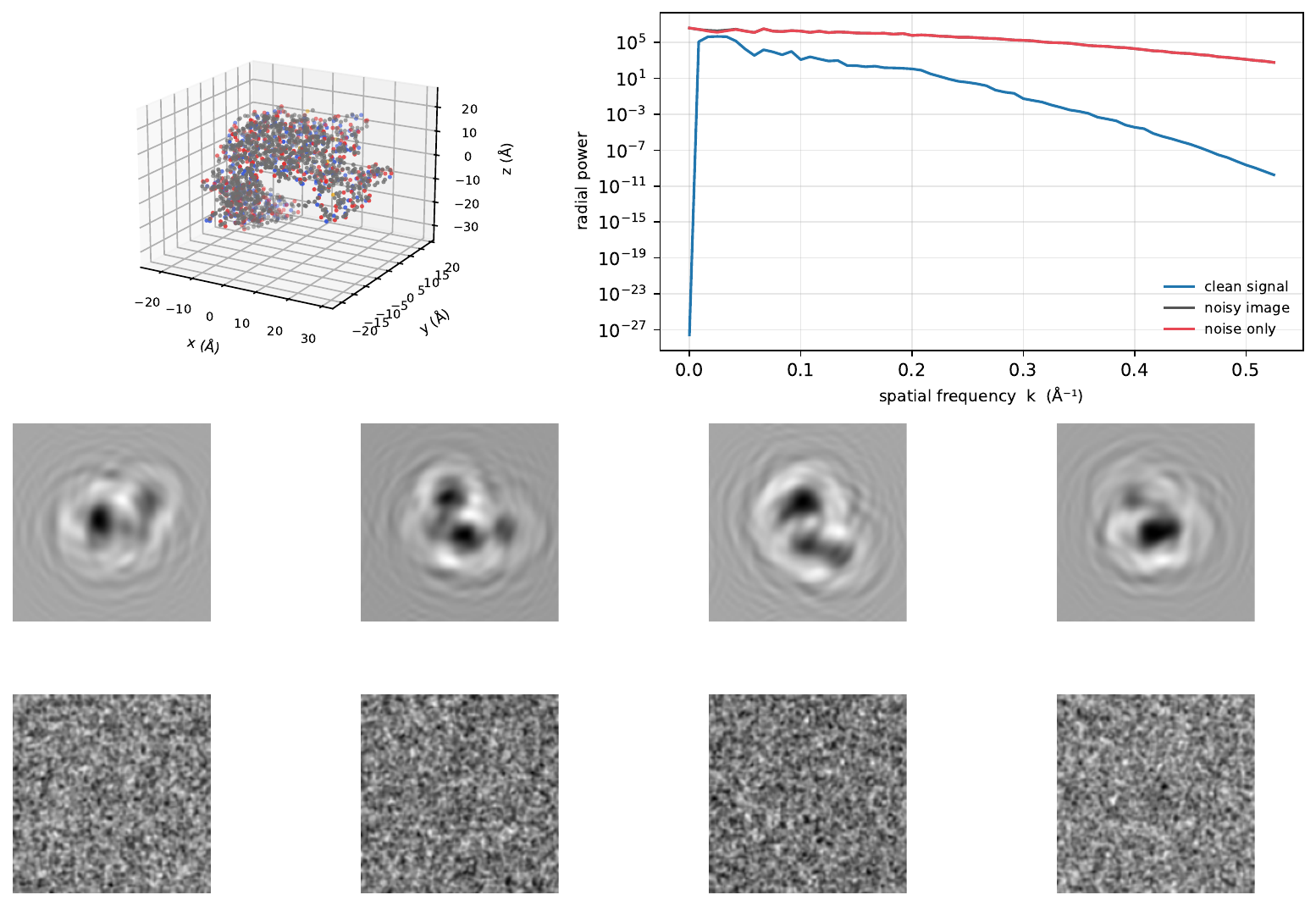}
    \caption{Image generation: the all-atom model is geometrically projected to the 2D detector, and a Doyle--Turner $4$-Gaussian model is applied to simulate the projected potential, after which a CTF with a Gaussian envelope is applied. Additive correlated Gaussian noise is added, which decays more slowly than the image envelope, as is clear from the power spectra. }
    \label{fig:imgen}
\end{figure}

\subsection{Unregularized homogeneous setting: Adenylate Kinase}\label{sec:exp_adk}
We first apply our gradient flow method for indirect matching of 3D protein backbone structures. As a first step, we illustrate the method on \emph{E. coli} adenylate kinase protein, which has a clear closed-to-open transition.

To construct the target, we apply the data generation of \cref{sec:datagen} to the final frame of a protein trajectory from \cite{Beckstein2018,Seyler2015}, which was obtained by forced molecular dynamics simulations. 
Dynamic importance sampling was used to generate a set of trajectories all undergoing a closed-to-open transition, so the final frame is in all cases the open conformation. 
We thus aim to reconstruct the open, final, conformation from synthetic data, generated as in \cref{sec:datagen}. 
From this single open endpoint we generate $\numImgs = 16$ images, each a projection at a known, but uniformly random orientation.
We set the signal-to-noise ratio to $0.01$.

To construct the template, we use the AlphaFold database entry corresponding to the closed conformation\footnote{\url{https://alphafold.ebi.ac.uk/entry/A0A7H9QZH8}} \cite{Jumper2021,Varadi2021}. 
We align the AlphaFold data  with the first frame of the protein trajectories molecular dynamics data using Procrustes analysis without scaling. 
Procrustes analysis  computes the orthogonal linear transformation aligning one point cloud with another that is optimal in the sense that it minimizes the \emph{Procrustes distance} \cite{Gower1975,Kendall1989}. 

For this transition, the deformation does not induce self-intersections, so we run the flow without regularization. 
To specify the forward model, all weights in \cref{eq:gauss} are set equal to $1$, and all standard deviations to $5$. 
We run the flow \cref{eq:gf_so3} discretized with the Lie--Euler method (see for instance the survey by Iserles et al. \cite{Iserles2000}) with a step size of $
    5/\numImgs.
$
 The $1/\numImgs$ scaling compensates for the data-fidelity gradient growing with the number of images.
As we use $\numImgs = 32$, the step size is $0.041667$, and we run the flow for a total flow time of $T = 100$ for a total of $2400$ steps. 
The results are shown in \cref{fig:adk_res}. 
In \cref{fig:adk_res_res}, it is clear that we manage to reconstruct the open conformation, by deforming the closed template. 
The flow convergence behavior is shown in \cref{fig:adk_res_flow}, where we see that the flow behaves like a typical gradient flow, with early rapid convergence that slows down as the gradient norm decreases. 
\begin{figure}[htbp]
  \centering
  \begin{subfigure}[b]{\textwidth}
    \includegraphics[width=\linewidth]{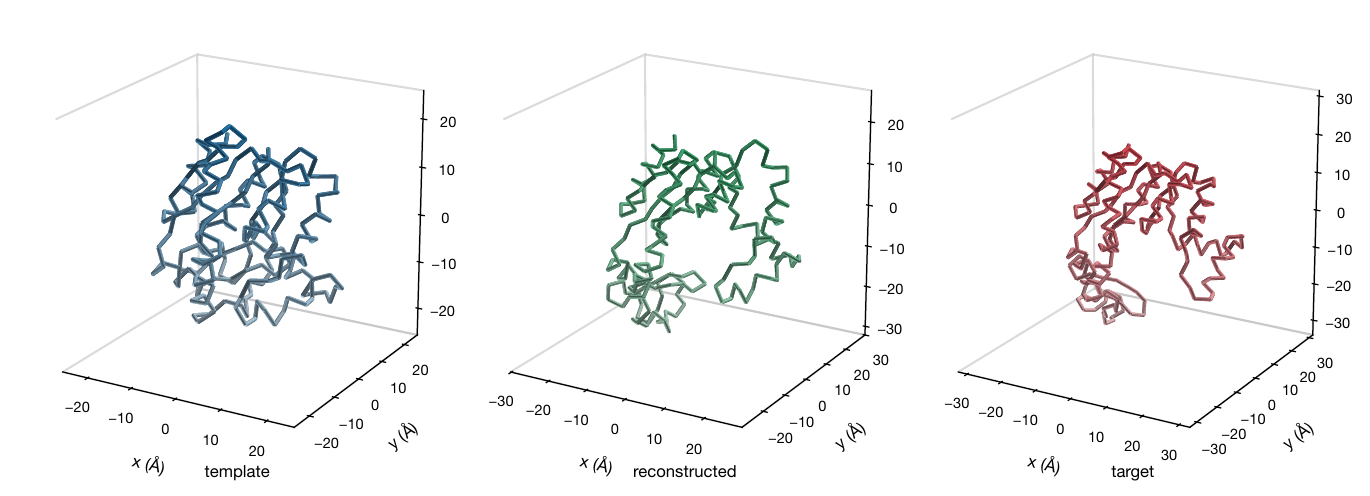}
    \caption{The blue template (left) is deformed into the green structure (middle), to match the indirectly observed structure.}
    \label{fig:adk_res_res}
  \end{subfigure}
  \hfill
  \begin{subfigure}[b]{\textwidth}
    \includegraphics[width=\linewidth]{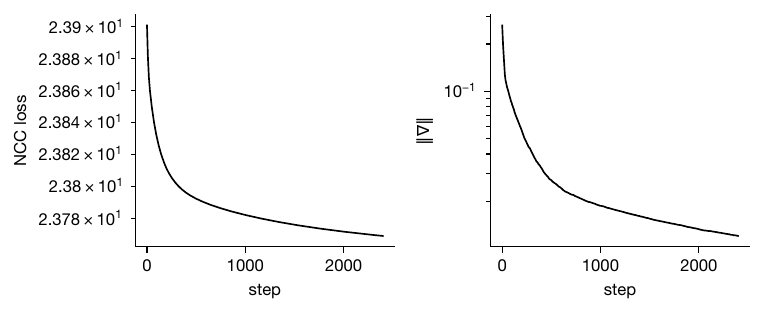}
    \caption{Empirical convergence behavior of the flow. After an initial rapid convergence, the flow slows down as the gradient norm decreases. }
    \label{fig:adk_res_flow}
  \end{subfigure}
  \caption{Results for the adenylate kinase protein.  The gradient flow is unregularized, that is, there is no penalty to discourage for instance self-intersections.}
  \label{fig:adk_res}
\end{figure}
While the resulting deformed template is not a perfect reconstruction, it is clear from \cref{fig:adk_res} that the gradient flow manages to recover the open conformation from the noisy observations.

\subsection{Regularized multi-chain homogeneous setting: Respiratory syncytial virus F protein}
\label{sec:exp_multi}

As a second example, we consider the pre-fusion to post-fusion transition of the RSV F glycoprotein.
It has two stable conformations, a pre-fusion closed shape resembling a closed flower bud, and a post-fusion conformation where the chains have shot out, resulting in a needle-like structure. 
The protein consists of three chains, made up of $350$, $347$ and $342$ residues respectively, for a total of $1039$ residues. 

The target is generated as in \cref{sec:datagen} from the post-fusion all-atom structure, with
$\numImgs = 32$ images at known, uniformly random orientations and a signal-to-noise ratio of $0.05$. The template is the pre-fusion backbone, obtained from McLellan et al. \cite{McLellan2013} \footnote{Protein Data Bank entry 4MMS}, and the flow deforms it towards the post-fusion conformation, obtained from \cite{Swanson2011} \footnote{Protein Data Bank entry 3RRR}.

For a rearrangement of this size the unregularized flow produces backbone self-intersections, since the data fidelity
constrains only the projected images and not the three-dimensional geometry, so this example also serves to illustrate how a regularization can be added.
We add an excluded-volume penalty on the reconstructed \ch{C_{$\alpha$}} coordinates,
\begin{align*}
    \tilde{\RegFunc}(\bbpos) = \frac{\lambda}{2}\!\!
    \sum_{\substack{i,j \\ \lvert i-j\rvert > 2\\ r_{ij} < R_0}}
    \!\!(R_0 - r_{ij})^2 ,
\end{align*}
where $\lambda > 0$ is the regularization strength, $r_{ij} = \lVert \bbpos_i - \bbpos_j \rVert$, and $R_0 = 5\,\text{\AA}$ is the cutoff. The sum runs over pairs that are not close along the chain, $\lvert i - j \rvert > 2$. 
Nearby pairs are excluded because consecutive \ch{C_{$\alpha$}} atoms are close in space.

The penalty enters the flow through the same construction as the data fidelity. We obtain an additional term in the energy by composing the group action on the template, the reconstruction map to absolute coordinates, and the penalty $\tilde{\RegFunc}$,
\begin{align*}
    \RegFunc(g) = \tilde{\RegFunc}\bigl(\opStyle{M}(g \cdot \template)\bigr) .
\end{align*}
Its gradient with respect to $g$ follows from the chain rule and momentum map as in \cref{th:gradient_general} in complete analogy with how the gradient is computed for the data term, with $\tilde{\RegFunc}$ in place of the data fidelity, so the regularizer stays within the same variational framework, and the flow descends it as in \cref{eq:gf_so3}.

To choose the penalty strength we ran the reconstruction across a range of $\lambda$.
We again integrate \cref{eq:gf_so3} with the Lie--Euler method with a step size of  $6.4975\times 10^{-4}$, and we run for a total time of $T = 1000$.
\Cref{tab:rsvf_reg} reports, for each $\lambda$, the number of \ch{C_{$\alpha$}} pairs separated by more than two residues along the chain  which are closer than $R_0$, together with the minimum separation and the counts below $4$ and $3\,\text{\AA}$. 
Without regularization the reconstruction contains $698$ pairs closer than $3\,\text{\AA}$, with a minimum separation of $0.27\,\text{\AA}$: the backbone self-intersects severely.
Any nonzero $\lambda$ removes the close contacts, and from $\lambda = 10^{-3}$ upward no pair lies below $4\,\text{\AA}$. 
The reconstructions are essentially unchanged over the interval $\lambda \in [10^{-3}, 10^{-1}]$ (\Cref{fig:rsvf_reg}), so the result is insensitive to the precise value.
However, at $\lambda = 1$ the penalty dominates the data term and the structure is visibly distorted. 
The results with $\lambda = 10^{-3}$ are shown in \cref{fig:rsvf_res}. The flow recovers the post-fusion conformation from the
pre-fusion template, but again, not yielding perfect results.

\begin{figure}[htbp]
    \centering
    \begin{subfigure}[b]{\textwidth}
        \centering
        \includegraphics[width=\linewidth]{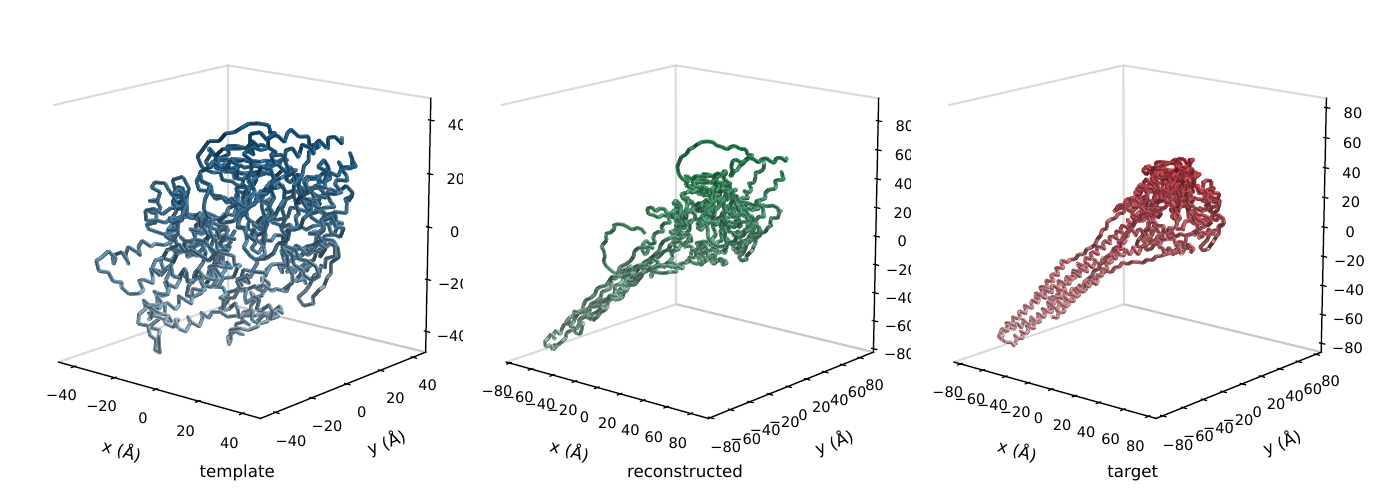}
        \caption{The pre-fusion template (left) is deformed (middle) to match the
    indirectly observed post-fusion structure (right).}
        \label{fig:rsvf_res_res}
    \end{subfigure}
    \hfill
    \begin{subfigure}[b]{\textwidth}
        \centering
        \includegraphics[width=\linewidth]{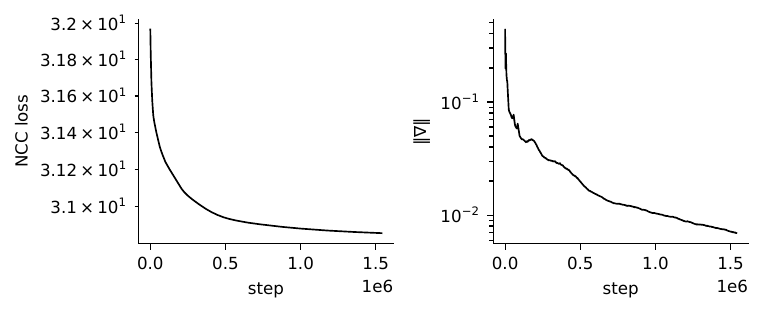}
        \caption{Convergence behaviour of the flow.}
        \label{fig:rsvf_res_flow}
    \end{subfigure}
    \caption{Results for the RSV F glycoprotein. The flow is regularised with the
  excluded-volume penalty~$\tilde{\RegFunc}$ to discourage backbone self-intersections during the
  large pre- to post-fusion rearrangement.}
    \label{fig:rsvf_res}
\end{figure}
\begin{table}[htbp]
  \centering
  \begin{tblr}{
      colspec = {l S[table-format=4.0] S[table-format=1.2] S[table-format=4.0] S[table-format=3.0]},
      row{1} = {guard, font=\bfseries},
    }
    \toprule
    $\lambda$ & {pairs $<R_0$} & {$\min r$ (\AA)} & {$r<4$\,\AA} & {$r<3$\,\AA} \\
    \midrule
    $0$          & 2570 & 0.27 & 1496 & 698 \\
    $10^{-4}$    &  872 & 3.36 &   63 &   0 \\
    $10^{-3}$    &  537 & 4.72 &    0 &   0 \\
    $10^{-2}$    &  486 & 4.96 &    0 &   0 \\
    $10^{-1}$    &  473 & 5.00 &    0 &   0 \\
    $10^{0}$     &  113 & 5.00 &    0 &   0 \\
    \bottomrule
  \end{tblr}
  \caption{Number of \ch{C_{$\alpha$}} pairs closer than $R_0 = 5\,\text{\AA}$ in the final RSV F reconstruction, for a range of regularization strengths $\lambda$. Pairs within two residues in sequence are excluded, since consecutive \ch{C_{$\alpha$}} atoms tend to lie $\approx 3.8\,\text{\AA}$ apart.
The cutoff $R_0$ is the range of the excluded-volume penalty, so the first column counts the pairs on which the penalty is active. The
$4\,\text{\AA}$ and $3\,\text{\AA}$ columns report how many of these pairs are closer still.}
  \label{tab:rsvf_reg}
\end{table}
\begin{figure}[htbp]
    \centering
    \includegraphics[width=\linewidth]{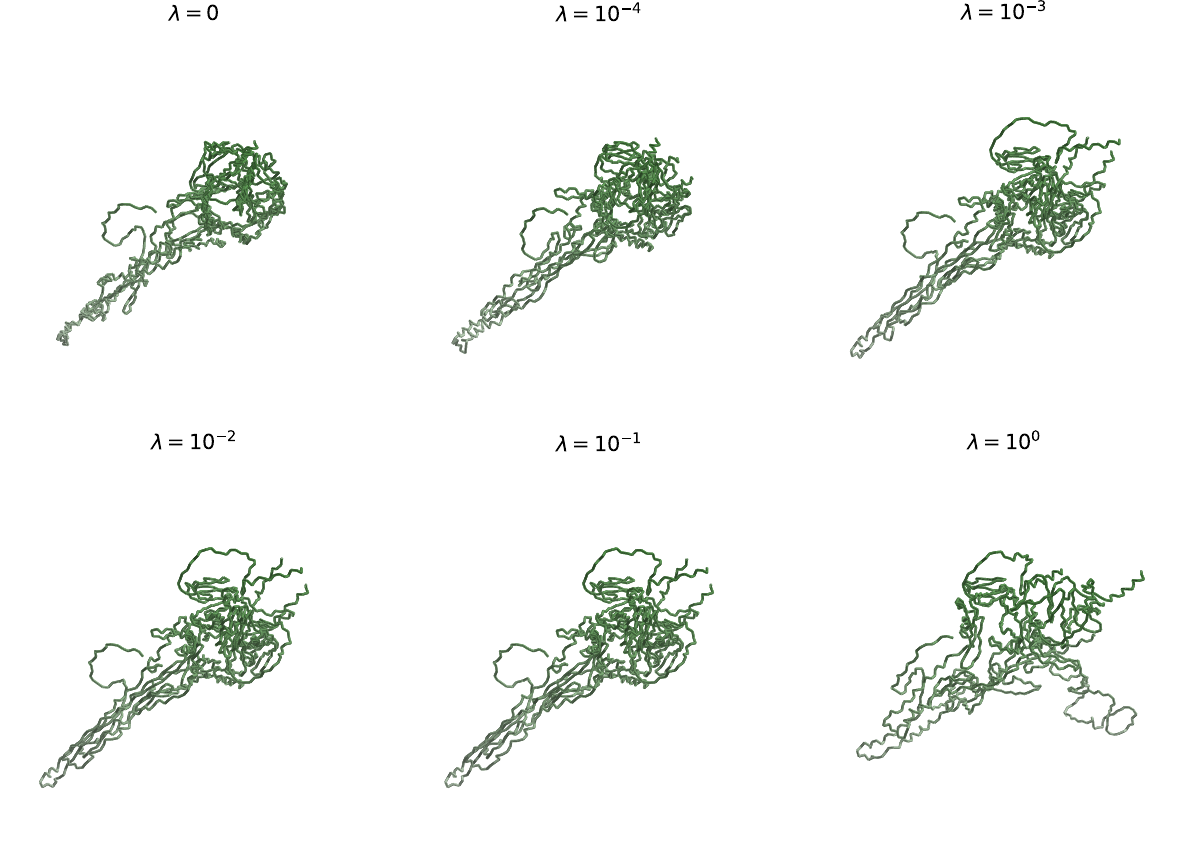}
    \caption{Reconstructions of the RSV F protein across regularization strengths $\lambda$. At $\lambda = 0$, the protein folds through itself, but between $\lambda = 10^{-3}$ and $\lambda = 10^{-1}$, such clashes do not appear. When $\lambda = 1$ the penalty overwhelms the data term and the structure is distended.}
    \label{fig:rsvf_reg}
\end{figure}

\subsection{Capturing a conformational transition}
\label{sec:exp_het}
The previous examples recover a single conformation from the images of one structure. We now consider a data set containing a conformational transition, and ask whether the gradient-flow based method can recover it.
We again use \emph{E. coli} adenylate kinase, and the dynamic importance sampling trajectories of \cite{Beckstein2018,Seyler2015}, but now we use every frame rather than only the final one. We take $80$ trajectories and generate one image per frame at a known, uniformly random orientation, giving $\numImgs = 3946$ images of $3946$ distinct
conformations, spanning the closed-to-open transition. The signal-to-noise ratio is $0.05$. The image generation procedure used is the same as that in \cref{sec:exp_adk}.

Reconstructing one structure from all images would result in an averaged conformation.
We instead group the images by a coordinate estimated from data, and reconstruct each group separately. We emphasize that estimating such a coordinate is a problem in its own right and is not the contribution here. We use a deliberately simple choice, in order to demonstrate that the gradient flow method can be applied unchanged once a coordinate is available. 

We take a trajectory that was not used to generate images and extract its first and final frame. 
For each image we score both references
at the known orientation of that image, and take the difference
\begin{align*}
    t_m = \mathrm{NCC}(\dataother^{\mathrm{closed}}_m, \data_m)
        - \mathrm{NCC}(\dataother^{\mathrm{open}}_m, \data_m),
\end{align*}
where $\dataother^{\mathrm{closed}}_m$ and $\dataother^{\mathrm{open}}_m$ denote the two references
projected at the orientation of image $\data_m$. 
The idea is that $t_m$ decreases as the chain transitions from closed to open. \Cref{fig:het_coord} shows that the coordinate approximately tracks the transition, but with a wide variation due to the noise level. 
\begin{figure}[htbp]
    \centering
    \includegraphics[width=0.55\linewidth]{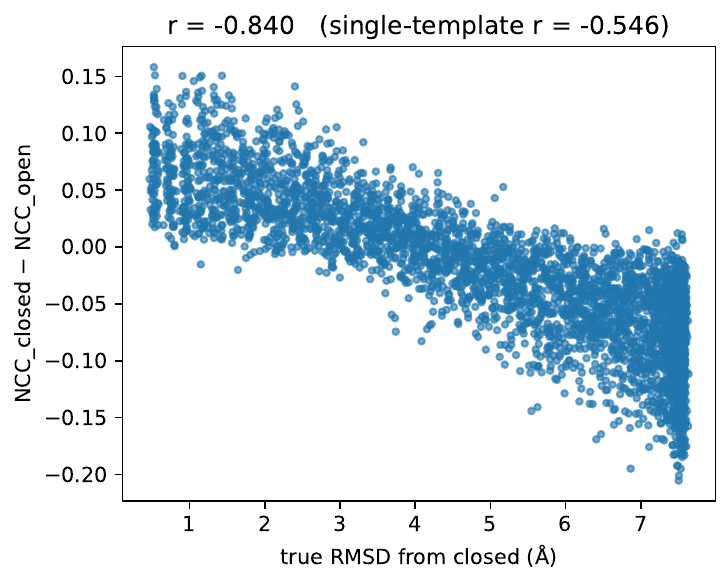}
    \caption{The conformational coordinate $t_m$ against the true distance of each
    imaged conformation from the closed reference. Each point is one of the
    $\numImgs$ images. The coordinate is computed from the noisy images and the two
    references only.}
    \label{fig:het_coord}
\end{figure}

We sort the images by $t_m$ and partition them into $50$ bins, each containing $78$ images.
 Each group is reconstructed independently, by running the gradient flow in \cref{eq:gf_so3} starting from the closed template.
We re-iterate that the open endpoint is only used in constructing the coordinate $t_m$, and is not used in the reconstruction. 
We integrate with the Lie--Euler method at a step size of $0.0133$ to a total time of $T = 100$,
for $7503$ steps per group.

The results are shown in \cref{fig:het_backbones} and \cref{fig:het_disp}. 
The reconstructions progress from the closed to the open conformation. 
The conformations recovered are different from the template. Indeed, the disparity between the  mean conformations in each bin and the reconstructions stay approximately constant across the transition, whereas the disparity between the template and the bin means grows. 
We reiterate that  mean conformations are not used except for evaluation and thus do not enter the gradient flow or the coordinate.

\begin{figure}[htbp]
    \centering
    \includegraphics[width=\linewidth]{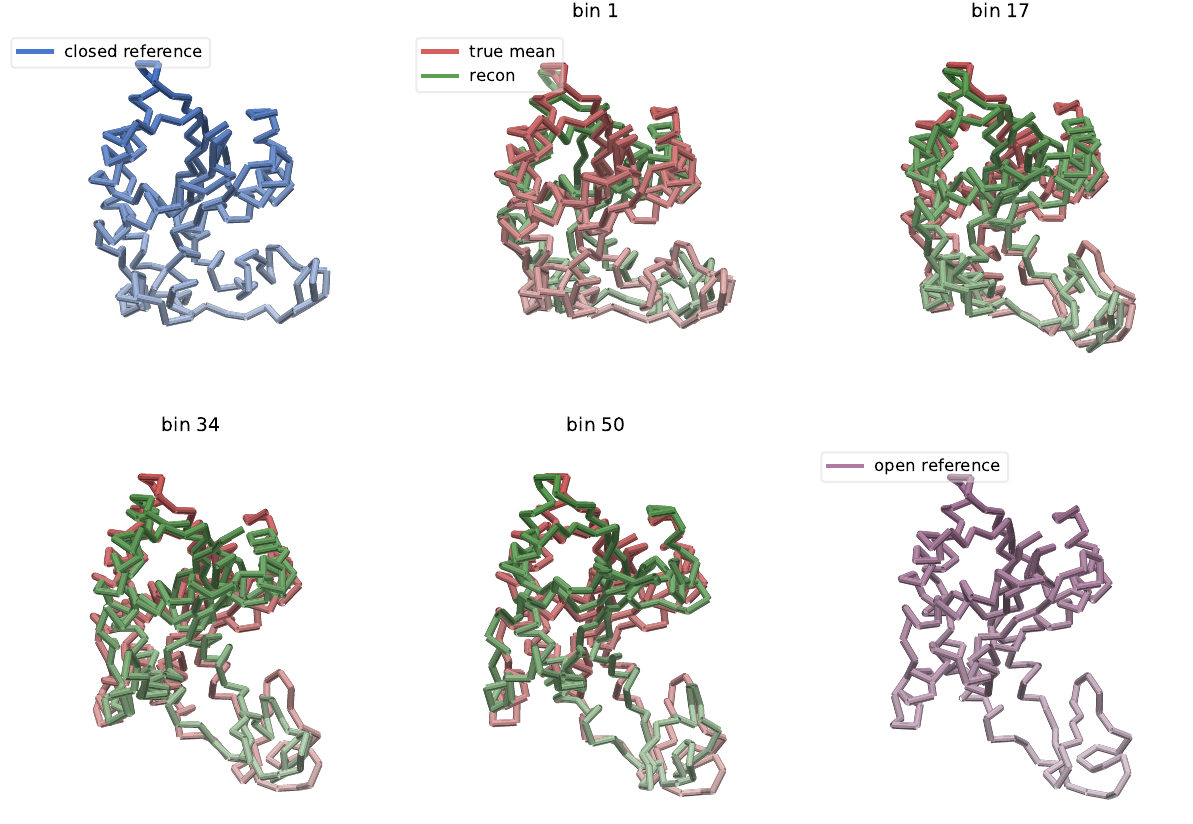}
    \caption{Reconstructions across the conformational coordinate. Green: the reconstruction for that group; red: the group's true mean conformation, shown for evaluation only. The flow starts from the closed reference (blue) and is never given the open reference (purple).}
    \label{fig:het_backbones}
\end{figure}

\begin{figure}[htbp]
    \centering
    \includegraphics[width=0.55\linewidth]{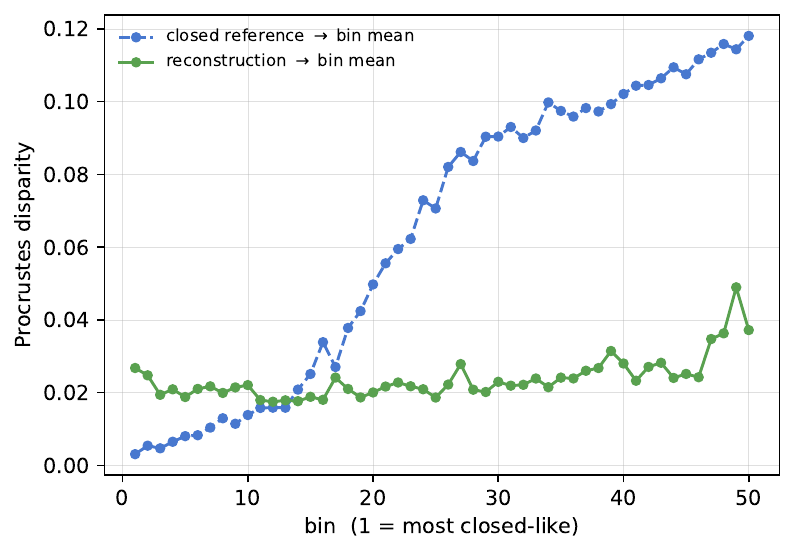}
    \caption{Procrustes disparity between each group's true mean conformation and,
    respectively, the closed reference and the reconstruction. The reconstruction  holds approximately constant disparity across the transition while the fixed reference degrades.}
    \label{fig:het_disp}
\end{figure}

\subsection{ Parameter sensitivity study}
The preceding examples were qualitative, in the sense that we did not examine how the reconstruction quality depends on the parameters of the problem. We now do so on a controlled synthetic problem: a single adenylate kinase conformation, reconstructed from its projections as in \cref{sec:exp_adk}, so that the true structure is known, and the
reconstruction error can be measured directly.
As the error measure, we use the rotation and translation Procrustes disparity between the reconstruction and the true structure, a dimensionless quantity that is invariant to rigid motion. 
For each parameter setting, we run $50$ independent trials with fresh noise and orientations, and report the median and inter-quartile range. 
A trial is counted as converged when the relative change in the loss falls below $10^{-4}$; trials that reach the step cap without meeting this criterion are drawn with open markers.
We vary three quantities: the number of projections $\numImgs$, the noise level, and the integration step size. 

\Cref{fig:param_M} shows the disparity against the number of projections $\numImgs$, at three signal-to-noise ratios. 
At every noise level, the disparity decreases monotonically with $\numImgs$ and then flattens onto a common floor of about $0.025$ for $\numImgs \ge 16$.
This floor might be due to the model mismatch, the coarse-graining of the proteins used in the reconstruction (as compared to the all-atom model used to generate data).
\begin{figure}[htbp]
    \centering
    \includegraphics[width=0.6\linewidth]{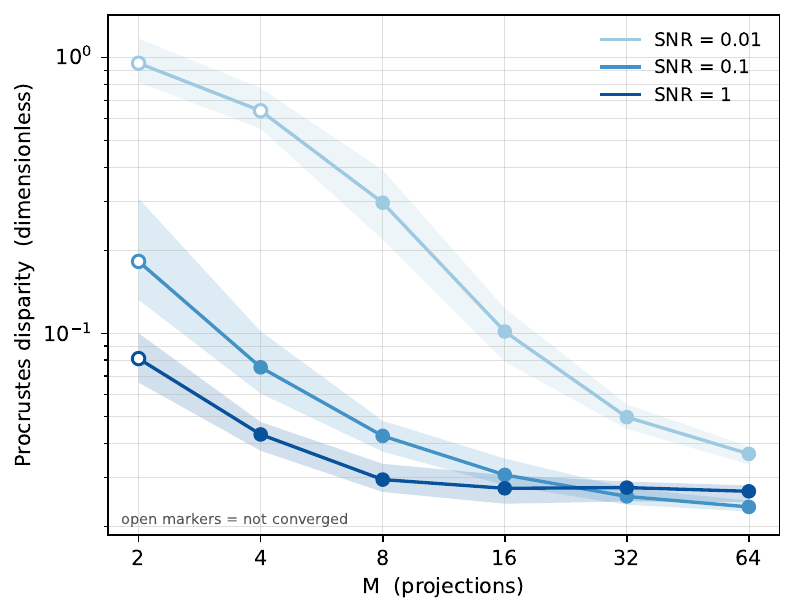}
    \caption{Procrustes disparity to the true structure against the number of projections $\numImgs$, at three signal-to-noise ratios. Median and inter-quartile range over $50$ trials. Open markers denote settings where not all trials converged. The disparity decreases with $\numImgs$ and flattens onto a model-mismatch floor.}
    \label{fig:param_M}
\end{figure}

\Cref{fig:param_noise} shows the disparity against the noise level, at 
three values of $\numImgs$. 
For $\numImgs = 64$ the disparity remains low and nearly flat across the whole range, degrading only at the highest noise. The disparity starts to grow again at a small nonzero noise level as the noise level further decreases. 
This is possibly due to that a small amount of noise suppresses some of the errors caused by details that the coarse forward model cannot represent. 
More projections therefore buy robustness to noise, consistent with \cref{fig:param_M}.

\begin{figure}[htbp]
    \centering
    \includegraphics[width=0.6\linewidth]{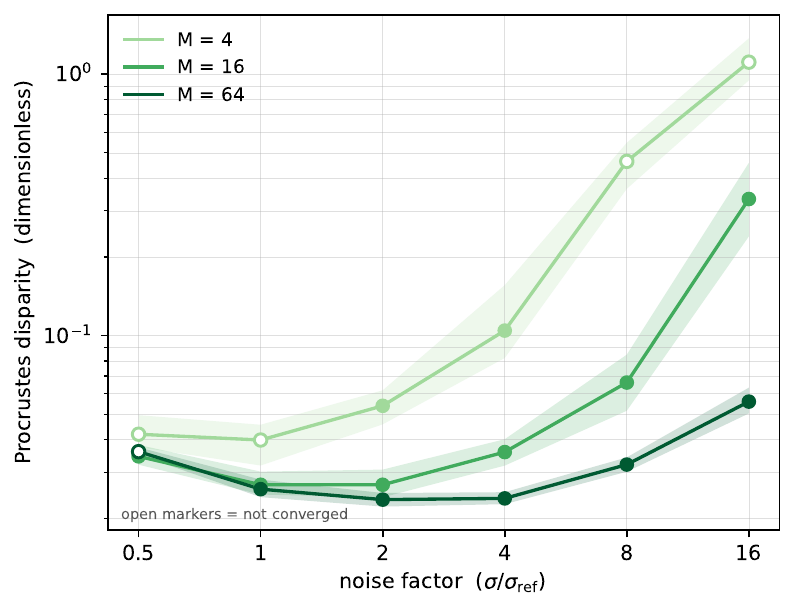}
    \caption{Procrustes disparity to the true structure against the noise level, at three numbers of projections $\numImgs$. More projections give greater robustness to noise.}
    \label{fig:param_noise}
\end{figure}

Finally, \cref{fig:param_dt} examines the integration step size, reported as a multiple of the calibrated step used in \cref{sec:exp_adk}. The disparity is essentially constant for step sizes up to the calibrated value, and the flow diverges only once the step is increased several-fold beyond it. The method is therefore insensitive to the step size within a broad stable range, and the value used throughout this work lies inside that range. 
Note also that if the step size is too large, the flow may not converge. 

\begin{figure}[htbp]
    \centering
    \includegraphics[width=0.6\linewidth]{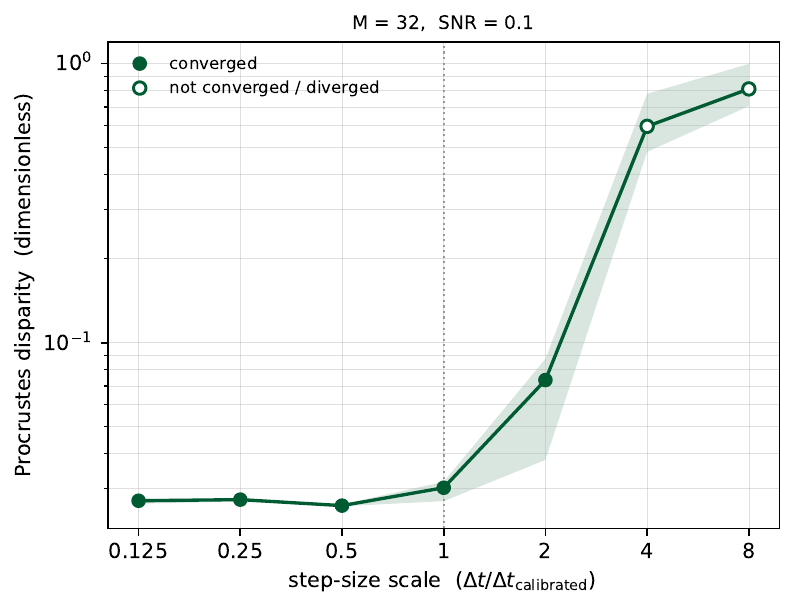}
    \caption{Procrustes disparity against the integration step size, as a multiple of the calibrated step. The reconstruction is insensitive to the step within a broad range and diverges only well beyond the calibrated value.}
    \label{fig:param_dt}
\end{figure}

\section{Conclusion and outlook}
We have presented a method for recovering the backbone of a protein directly from \ac{Cryo-SPA} data, without first reconstructing a 3D map. 
The reconstruction is formulated as an indirect shape matching problem: a point-cloud template is deformed by a gradient flow on a Lie group of deformations to match projections against the observed data. 

We derived the gradient flow in its general geometric form, so that the framework can be applied to a wider range of inverse problems by varying the group. For instance, tomography is a possible application, and requires the use of infinite-dimensional groups. In this setting, questions such as existence of the flow are non-trivial to answer, and are something we intend to explore in future work. 

In the \ac{Cryo-EM} setting, we derived closed-form expressions of the gradient flow and demonstrated it on several synthetic examples, both single- and multi-chain backbones as well as on images of conformational changes, where we could recover the transition given a sufficiently descriptive coordinate. 

Several assumptions limit the present study. The orientations are taken to be known but in practice they must be estimated, and jointly with the conformation this is a substantially harder problem. In the heterogeneous example, the conformational coordinate is supplied by two reference structures and estimated separately from the reconstruction, but recovering such a coordinate from the data alone is a problem in its own right and is not addressed here.

These limitations indicate the natural next steps. The most immediate is to relax the known-orientation assumption and estimate poses jointly with the deformation. Here, we could for instance apply the work of Diepeveen et al. \cite{Diepeveen:2023aa}. 
 Finally, applying the method to experimental data, and replacing the fixed conformational coordinate with one learned from the images, would move it from a controlled demonstration towards a practical reconstruction tool.

\section*{Acknowledgments}
The authors thank Klas Modin for valuable discussion and insights. EJ was funded by the Knut and Alice Wallenberg Foundation grant 2024.0440. JK and OÖ acknowledge support from the Swedish Research Council grant 2020-03107.

\bibliographystyle{plain}
\bibliography{refs}

\end{document}